\documentclass[aps,twocolumn,pra,superscriptaddress,nofootinbib]{revtex4-2}

\usepackage{amsmath}
\usepackage{amssymb,nicefrac}
\usepackage{mathrsfs}
\usepackage{bm,dsfont}
\usepackage[table,dvipsnames]{xcolor}
\usepackage{enumitem}
\usepackage{amsthm}
\usepackage{graphicx}
\graphicspath{{arxiv_update/}{./}}
\usepackage{verbatim}
\usepackage{physics}
\usepackage[normalem]{ulem}
\usepackage[colorlinks=true,linkcolor=blue,citecolor=blue,urlcolor=blue]{hyperref}
\usepackage{hypcap}
\usepackage{cleveref}

\newcommand{\id}{\mathds{1}}
\newcommand{\ii}{\mathrm{i}}

\newtheorem{proposition}{Proposition}

\begin{document}

\title{Quantum models of interaction Hamiltonians and their paradoxes}

\author{Nicolas Gisin}
\affiliation{Department of Applied Physics, University of Geneva, Switzerland}
\affiliation{Constructor University, 28759 Bremen, Germany}

\author{Serge Massar}
\affiliation{Laboratoire d'Information Quantique, Universit\'e libre de Bruxelles (ULB), Belgium}

\author{Jef Pauwels}
\email{jef.pauwels@unige.ch}
\affiliation{Department of Applied Physics, University of Geneva, Switzerland}
\affiliation{Constructor University, 28759 Bremen, Germany}

\author{Pavel Sekatski}
\affiliation{Department of Applied Physics, University of Geneva, Switzerland}

\begin{abstract}
In quantum physics it is commonplace to model the interaction of remote systems with a many-body Hamiltonian. Taking such an action-at-a-distance description {\it \`a la lettre} leads to various paradoxes related to faster-than-light communication and apparent inconsistencies in local energy accounting. It also neglects residual effects, such as entanglement between the remote systems and the mediator that implements the interaction, or the decoherence that arises when the remote systems undergo local evolution. We study simple microscopic quantum models that respect the light cone by design and reproduce two-body Hamiltonians. For these models we quantitatively analyze the residual effects of the microscopic mediator on the remote systems, including dressing of stationary states, and decoherence in the presence of fast local control. We show how the models resolve the paradoxes.
\end{abstract}

\maketitle

\section{Introductory paradoxes}

Ever since Newton, it has often been useful to describe the interaction of distant bodies through an instantaneous potential: a single function (the Hamiltonian) of the remote degrees of freedom that directly generates the forces. In modern physics this viewpoint is understood as an effective description. Fundamentally, interactions are mediated by additional degrees of freedom (fields, modes) that respect the principle of local action and propagate disturbances at finite speed. 

Effective nonlocal Hamiltonians are operationally very powerful. They describe a vast range of phenomena, from planetary orbits to atomic spectra. In quantum science, they have been used to quantify how interactions generate entanglement and classical communication~\cite{DurVidalCiracEtAl2001,BennettEtAl2002Sim,BennettHarrowLeungSmolin2003Capacities,GisinZambrini2018Clocks,NielsenEtAl2003Resource}, and they give rise to protocols in which nonlocal observables become accessible given additional resources or classical communication~\cite{Teleportation1993,GroismanVaidman2001,GroismanReznikVaidman2003,Vaidman2003Instantaneous,ClarkEtAl2010,IshizakaHiroshima2008,PauwelsEtAl2025PRX}.

Taking an effective interaction Hamiltonian too literally has two negative consequences. First, it obscures the local propagation mechanism and can suggest instantaneous influence. Second, it obscures where interaction energy resides on short time-scales, so that rapid actions invite conflicting accounts of local energy flow. Both aspects connect to classic debates on time--energy tradeoffs and the operational meaning of rapid energy readout~\cite{LandauPeierls1931,AharonovBohm1961,AharonovReznik2000,AharonovMassarPopescu2002,MassarPopescu2005,AtiaAharonov2017,PaivaLoboCohen2022}. 

These difficulties arise because the effective interaction Hamiltonians are never truly fundamental: they emerge by eliminating local mediators.
 This is seen explicitly in many familiar systems, from spin--spin couplings in NMR~\cite{ErnstBook1987} and Born--Oppenheimer molecular forces~\cite{AtkinsFriedman2010} to circuit QED, where a near-resonant mode can be adiabatically removed~\cite{BlaisEtAl2004cQED} (see also~\cite{NielsenChuang2010}), and in light-mediated remote interactions~\cite{wang2017two,KargEtAl2019}.

If one nevertheless takes a genuinely nonlocal many-body Hamiltonian {\it \`a la lettre} as a microscopic law, one is led to sharp paradoxes. We briefly isolate these tensions below, before resolving them by embedding the interaction into an explicit microscopic mediator model, in the spirit of modern quantum information theory.

Consider two systems, $A$ and $B$, whose interaction is modeled by a nonlocal Hamiltonian $H_{AB}\neq H_A \otimes \id_B+ \id _A \otimes H_B$. This model permits faster-than-light influences. Let $\rho_{AB}$ denote the global initial state. Suppose that Alice has the possibility to apply an instantaneous unitary operation $U_A=e^{-\ii x A}$ on her system, depending on the value of a bit $x\in\{0,1\}$. After time $t$ the reduced state of Bob's system is given by
\begin{align}
    \rho_B(t|x) &= \tr_A e^{-\ii t H_{AB}} e^{-\ii x A} \rho_{AB} e^{\ii x A}  e^{\ii t H_{AB}} \\
     &= \tr_A e^{-\ii x \tilde A}  \left(e^{-\ii t H_{AB}}  \rho_{AB}   e^{\ii t H_{AB}}\right) e^{\ii x \tilde A} 
\end{align}
with $\tilde A := e^{-\ii t H_{AB}} A \, e^{\ii  t H_{AB}} = A- \ii t [H_{AB},A] + O(t ^2)$. 
Choosing the Hermitian operator $A$ such that $\tilde A$ acts nontrivially on system $B$ yields, for a suitable initial state $\rho_{AB}$ and Bob's measurement, a signaling protocol with $\rho_B(t|1)\neq\rho_B(t|0)$ for arbitrarily small positive $t$. This shows that the Hamiltonian model of the interaction is inaccurate in the short-time limit.

The same inconsistency with the principle of local action persists in lattice models, where many-body interactions are modeled by quasi-local Hamiltonians. Here, Lieb-Robinson bounds~\cite{nachtergaele2010lieb} guarantee that information \emph{effectively} propagates at a finite speed. Nevertheless, this is only true up to the presence of possible superluminal corrections. While such superluminal influences are suppressed exponentially for finite-dimensional systems with short-range interactions, this is not always the case for longer-range interactions~\cite{tran2021lieb} or lattices of continuous-variable (CV) systems~\cite{eisert2009supersonic}.

Another paradox which arises in effective models with nonlocal Hamiltonians is that they seem to assign inconsistent ``work'' to local actions \cite{GisinZambrini2018Clocks}. The point is easiest to see in the smallest nontrivial example. Consider two qubits $A$ and $B$ coupled by the (effective) interaction
\begin{equation}
H_{AB} \;=\; \frac{\mu}{2}\bigl(X_A X_B + Y_A Y_B\bigr), \qquad \mu>0,
\end{equation}
and prepared in the ground state
$\ket{\Psi^-}=\frac{1}{\sqrt2}\bigl(\ket{01}-
\ket{10}\bigr),
$
with energy $E_{\Psi^-} =\bra{\Psi^-}H_{AB}\ket{\Psi^-} = -\mu$. 
Now suppose that Alice applies a fast local $X$ gate. Since $[X_A,H_{AB}]\neq 0$, the total energy changes. In fact,
\begin{equation}
(X_A\otimes \id_B)\ket{\Psi^-}=\ket{\Phi^-} \Rightarrow
\Delta E_A := E_{\Phi^-}-E_{\Psi^-}= \mu .
\end{equation}
However, if Bob applies $X$ at the same time, then
\begin{equation}
(X_A\otimes X_B)\ket{\Psi^-}=-\ket{\Psi^-}
 \Rightarrow 
\Delta E_{AB}=0 .
\end{equation}
Thus, if one identifies the ``local energy cost'' of Alice's operation with the resulting change in the total effective energy, that cost appears to depend on whether Bob acted simultaneously. In that sense, the work associated with a local intervention is no longer locally well-defined. Taken at face value within the effective two-qubit model, this could even be interpreted as Alice's energy cost depending on Bob's distant choice, and hence as signaling. This is the tension emphasized in Ref.~\cite{GisinZambrini2018Clocks}: if a nonlocal interaction Hamiltonian is treated as fundamental, energy conservation, local interactions, and no-signaling 
cannot be reconciled.
That work explicitly raised the need for a quantum-information-inspired microscopic model of effective nonlocal Hamiltonians. 

In fact, neither paradox is specifically quantum. Newton already noticed that nonlocal gravitational forces lead to counterintuitive action at a distance at astronomical scales~\cite{mary1955}, which was only resolved by Einstein's theory of general relativity over two centuries later~\cite{einstein1916foundation}. As a more down-to-earth example, consider two point masses attached to the ends of a spring. Modeling their interaction with the  
usual elastic potential $V(x_A,x_B) = \frac{k}{2}(x_B - x_A - L)^2$ also leads to all the paradoxes discussed above, with fast local displacement of the masses playing the role of local unitary operations. The resolution is that $V$ describes only the slow, collective stretching mode. A microscopic local theory of elasticity~\cite{landau2012theory}
involves an entire tower of internal modes of the spring; a fast displacement of the boundary masses excites these modes, launching compression wavepackets that propagate at a finite speed and carry energy locally. 

Our goal here is to see how exactly the Hamiltonian interaction model breaks down in ``paradoxical'' circumstances. To do so, rather than using full-blown relativistic quantum field theory, we present and study a simple toy model of two-body interactions which exhibits the following features:
\begin{itemize}
\item[(i)] It recovers the usual interaction Hamiltonian in its appropriate regime.  
\item[(ii)] It is explicitly local and respects the light cone by design. 
\end{itemize}

The second feature contrasts with common circuit-QED effective descriptions, where $H_{AB}$ is recovered by coupling both systems to \emph{delocalized} bosonic modes and adiabatically eliminating the latter; see, e.g., Ref.~\cite{BlaisEtAl2004cQED}. Because the modes are delocalized, locality is not enforced explicitly, and this approach cannot be used to resolve the above paradoxes.
A broad class of microscopic models instead describes spatially separated systems as coupling locally to propagating quantum fields~\cite{shen2005coherent,pichler2015quantum,wang2017two,KargEtAl2019, Milburn1999,WangZanardi2002}. Such descriptions make the local system--field couplings explicit and thereby implement requirement~(ii) at the microscopic level. In many applications, however, the propagation delay is neglected through a Markov or zero-delay approximation. Other studies retain finite travel times and the resulting non-Markovian dynamics explicitly~\cite{Grimsmo2015,PichlerZoller2016,dinc2019exact}. The relation to requirement~(i) varies among these works. 
In particular, Refs.~\cite{Milburn1999,WangZanardi2002,KargEtAl2019} derive conditions under which repeated light--matter interactions generate a coherent Hamiltonian interaction between remote systems.

 Here we use a simple causal microscopic model~\cite{Milburn1999,WangZanardi2002,SpillerEtAl2006Qubus,vanLoockEtAl2008} as a starting point for a conceptual study of  microscopic models  of an effective interaction between spatially separated systems. We retain explicitly the finite mediator flight time.
In the next section we first analyse general consequences that follow from a finite mediator propagation time without committing to a microscopic model. Section~\ref{sec:single-boson} introduces the simple displacement (or qubus) model, and shows how it gives rise to an effective interaction Hamiltonian. Section~\ref{catalytic-main} looks at the catalytic mediation mechanism from a general perspective, and in particular distinguishes exact from weak mediation. 
Section \ref{sec:dressed-states} then uses these models to describe how dressed stationary states arise.
Section~\ref{sec:continuous-field} develops the continuous-field limit. Section~\ref{sec:localDecoh} studies the breakdown of the effective description under local operations. We then return to the two paradoxes in Sec.~\ref{sec:back-paradoxes} and conclude with open questions in Sec.~\ref{sec:outlook}.

\section{General consequences of finite mediator propagation time}
\label{sec:GenConsequences}

We begin by discussing some general consequences one expects to see when the Hamiltonian $H_{AB}$ is not understood as an instantaneous nonlocal action at a distance, but rather as emerging from a strictly local interaction with a mediator degree of freedom that carries the interaction between systems $A$ and $B$. For this general discussion, no particular microscopic realization of the mediator will be assumed. However, it is important that the mediator takes a finite time $T$ to travel from one system to the other. An effective interaction Hamiltonian $H_{AB}^{\rm eff}$ arises only after the mediator has been eliminated and is therefore expected to describe time scales 
long
compared with $T$. At shorter time scales, the microscopic description leaves observable traces.

\subsection{Coarse-graining over the flight time}

A simple model-independent estimate illustrates the 
separation of time scales
implicit in the effective description. If the mediator flight time prevents one from assigning the effective state to an instant more accurately than a window of width $T$, one may replace the pure state $\Psi(t)=\ketbra{\Psi(t)}_{AB}$ by its time average
\begin{align}\label{eq:time-average}
\bar \rho(t) &=\;\frac{1}{T}\int_{-T/2}^{T/2} e^{-\ii s H_{AB}}\Psi(t) e^{\ii s H_{AB}}\,ds. \\
& = \Psi(t) -\frac{T^2}{24}\,[H_{AB},[H_{AB},\Psi(t)]]\,+\,O(T^4),
\end{align} with purity
\begin{equation}
\mathrm{Tr}\big(\bar\rho(t)^2\big)
\;=\;1-\frac{T^2}{6}\,\mathrm{Var}_\Psi(H_{AB}) \,+\,O(T^4)
\label{Eq:purity}
\end{equation}
where 
\begin{equation}
\mathrm{Var}_\Psi(H_{AB}) = 
\bra{\Psi(t)} H_{AB}^2\ket{\Psi(t)} -\bra{\Psi(t)}  H_{AB}\ket{\Psi(t)} ^2 . 
\label{Eq:purity-B}
\end{equation}
Equation~(\ref{Eq:purity}) is a heuristic coarse-graining estimate, not a microscopic prediction or a universal lower bound. 

\subsection{Hybridization, dressed states, and correlations}
\label{sec:Hyb}

A related causal argument follows from the entanglement-generating ability of $H_{AB}$. Consider a situation where systems $A$ and $B$ are initially unentangled. Split the mediator into subsystems $M_{A}$ and $M_{B}$ that, during a sufficiently short time interval, can only interact with $A$ and $B$, respectively. Such evolution cannot increase entanglement across the $AM_{A}:BM_{ B}$ cut. During this short time, any entanglement generated between $A$ and $B$ 
according to the effective theory 
must therefore already be available across this microscopic cut\footnote{In general, this does not require entanglement across the $AB:M_AM_B$ cut, as mixed-state counterexamples are known~\cite{cubitt2003separable}.}.

The effective state $\ket{\Psi(t)}_{AB}$ should thus be regarded as a reduced description of an underlying joint state of the systems and mediator. Even when an effective Hamiltonian is accurate and the reduced state of $AB$ is nearly pure, the microscopic state may contain residual correlations with $M$. Stationary states of the effective theory can likewise correspond to correlated or entangled system--mediator states, commonly called hybridized or dressed states. The precise type and amount of correlation is model dependent: finite propagation time alone does not imply entanglement across the $AB:M$ bipartition.

\subsection{Decoherence under fast local evolution}
\label{subsec-Decoh}

If $A$ or $B$ is subject to fast local evolution, non-negligible on the mediator-flight timescale, the microscopic dressed state underlying the effective description $\ket{\Psi(t)}_{AB}$ can be perturbed. The outgoing mediator can then retain information about the system, producing decoherence after it is discarded. This conclusion is generic but not unavoidable: it vanishes when the local evolution preserves the relevant coupling observable.

This effect 
can already be seen in a classical setting. Consider two electric dipoles interacting through the electromagnetic (e.m.) field. The dipoles are initially static, in fixed orientations. At time $t=0$ one rotates the dipoles, the rotation taking a time $\tau$. This time-dependent charge configuration will give rise to the emission of e.m. waves. These carry information about the dipole orientations, corresponding in the quantum case to decoherence. Note that if only one of the dipoles is rotated, the radiation only carries information about that dipole. The amount of radiation emitted depends on the angle by which the dipoles are rotated, and also on whether $\tau$ is shorter or longer than the duration $T$ required for light to travel from one dipole to the other (sudden versus adiabatic evolution). 

\section{The simplest model}
\label{sec:single-boson}

Here we adapt and generalise the  geometric-phase ``qubus'' gate  ~\cite{Milburn1999,WangZanardi2002,SpillerEtAl2006Qubus,vanLoockEtAl2008}. Our emphasis in subsequent sections will be on the finite propagation time of the propagating mediator, allowing us to study the in-flight corrections hidden by the effective two-system Hamiltonian.
This model consists of an interaction mediated by a single bosonic mode (i.e. a harmonic oscillator) with position and momentum operators obeying $[x,p]=\ii$.
The bosonic mode interacts successively with the $A$ and $B$ systems, as described in Fig.~\ref{fig:loop-A}. The successive interactions are given by the local unitary transformations
\begin{align}
    U_{ AM}^{(1)}&=e^{\ii g A x}\nonumber &\qquad
    U_{ BM}^{(1)}&=e^{\ii g  B p}\nonumber\\
    U_{ AM}^{(2)}&=e^{-\ii g  A x}  &\qquad
     U_{ BM}^{(2)}&=e^{-\ii g  B p}\,
     \label{Eq:Interactions}
\end{align}
where $A$ and $B$ are Hermitian operators and $g$ is the local coupling strength\footnote{Here and below we use the same notation $A$ and $B$ to label the systems and their observables. The meaning is nevertheless always clear from the context.}.
\begin{figure}[!t]
    \centering
    \includegraphics[width=\columnwidth ]{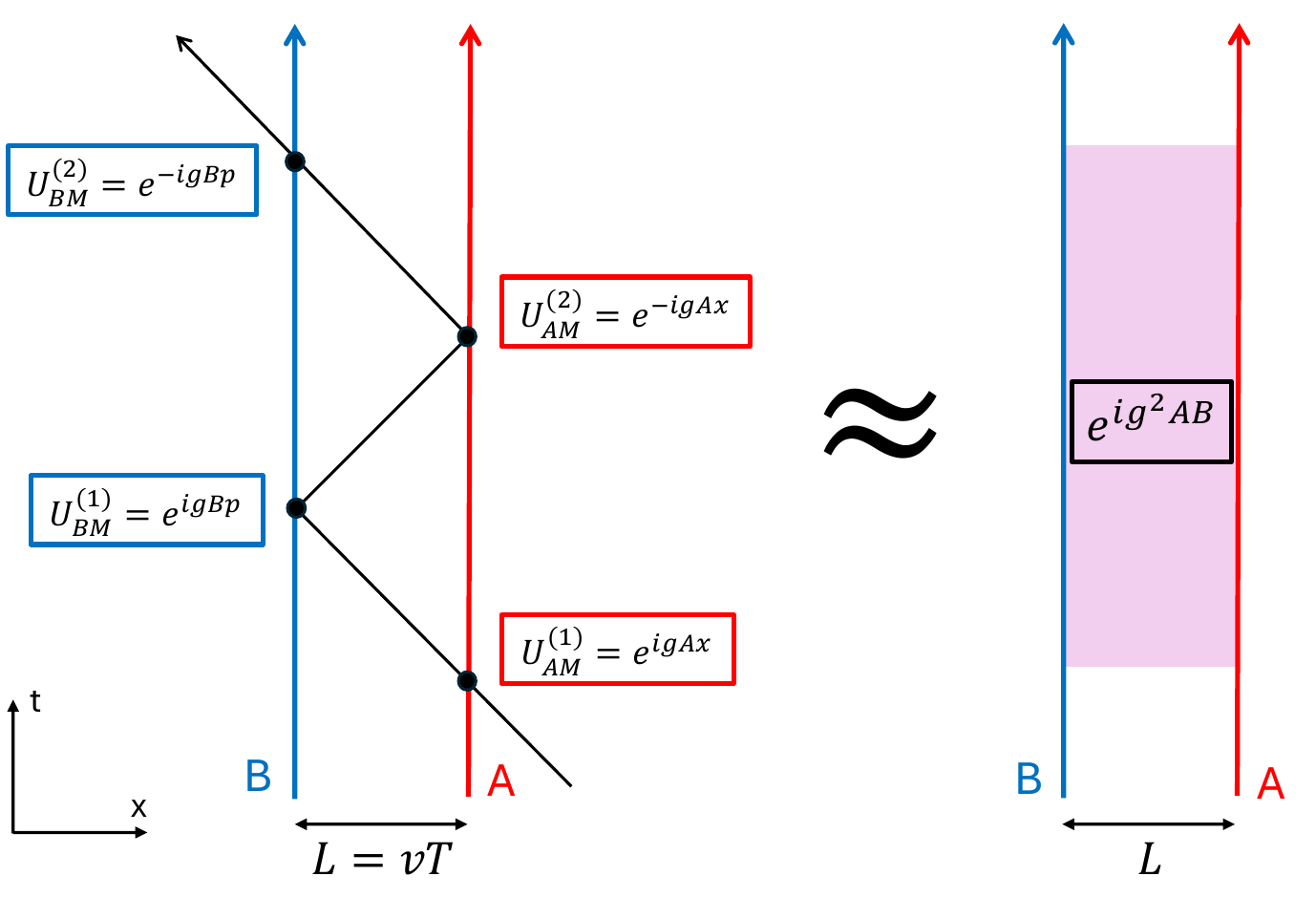}
    \caption{ Space-time diagram of a mediator interacting with systems $A$ and $B$, giving rise to an interaction $e^{\ii g^2 A B} \otimes \id_M$. Systems $A$ and $B$ are separated by distance $L$. The mediator travels at speed $v$, hence takes time $T=L/v$ to go from one system to the other. While the mediator trajectory is well defined in space-time, the effective interaction arises after the mediator has finished all its interactions and cannot be localized exactly in space-time.}
    \label{fig:loop-A}
\end{figure}

The overall interaction is therefore given by the unitary transformation
\begin{equation} \label{eq:twoloopprop}
    V_{ABM} = U_{BM}^{(2)} U_{AM}^{(2)} U_{BM}^{(1)}U_{AM}^{(1)} .
\end{equation}
Using the Weyl canonical commutation relation, we have
\begin{equation}
    V_{ABM} = e^{\ii g^2 A B} \otimes \id_M \, .
    \label{eq:VABM}
\end{equation}
After the four local interactions, the mediator has fulfilled its role of creating a nontrivial coupling between systems $A$ and $B$. Local causality is respected because each interaction is separated from the previous one by the time $T=L/v$ required for the mediator to travel between $A$ and $B$, where $L$ is the distance between $A$ and $B$ and $v$ is the mediator velocity.

The notion of interaction rate and effective interaction Hamiltonian arises if we suppose that new mediators arrive every $\tau$, see 
Fig.~\ref{Fig:ABM-FF} where we have chosen $\tau=4T$ in the left panel, and $\tau \ll  T$ in the right panel.
The interaction in Eq.~(\ref{eq:VABM}) then corresponds to the effective Hamiltonian
\begin{equation}
    H_{AB}^{\rm eff} =  - \frac{g^2}{\tau} A B = -\mu A B ,
\end{equation}  
where we introduced the effective coupling strength 
\begin{equation}
\mu = \frac{g^2}{\tau}  \ .
\end{equation}  

\begin{figure}
 \includegraphics[width=\columnwidth]{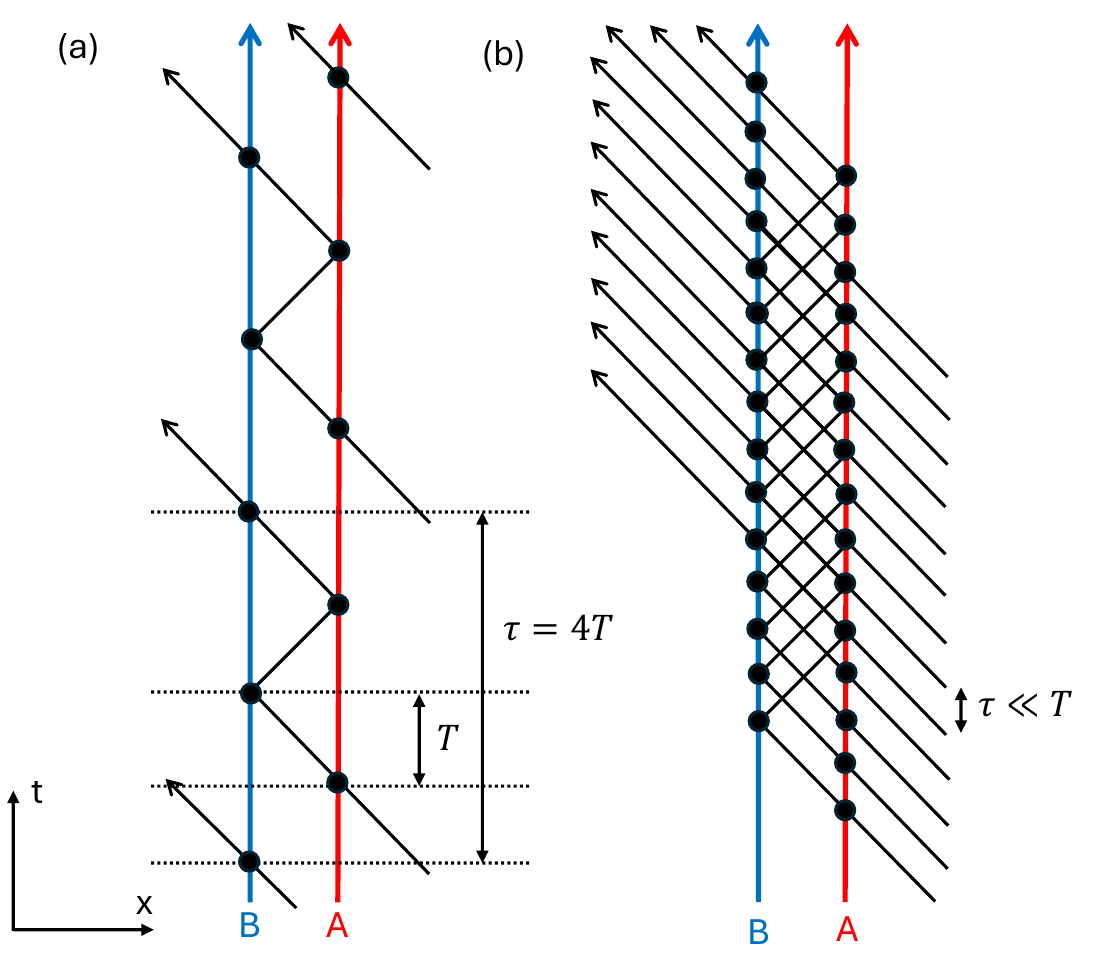}
\caption{Space-time diagrams of mediators (in black) continuously interacting with systems $A$ (red) and $B$ (blue). The interactions of the mediators with the $A$ and $B$ systems are shown by black dots. The mediators take time $T$ to travel from $A$ to $B$ (or back). A new mediator arrives every $\tau$. In panel (a) $\tau = 4T$ while in panel (b) $\tau \ll T$. In the limit $\tau \to 0$ we obtain the continuous field model.
Both models give rise to an effective interaction Hamiltonian $H^{\rm eff}_{AB}$.
}
\label{Fig:ABM-FF}
\end{figure}

\section{Catalytic mediation} \label{catalytic-main}
\label{sec:itsimei}

The simple model of Fig.~\ref{fig:loop-A} has three  crucial properties: it returns the mediator exactly to its initial state, it does so for every mediator state and for all  coupling strengths $g$, and it generates a nontrivial interaction between $A$ and $B$.  We now separate these properties and determine which features of the construction are responsible for them.  This will explain why the controlled-displacement model is exact, why the Jaynes-Cummings construction introduced below is only perturbative, and why overlapping mediator loops require additional care.

We consider two finite-dimensional systems $A$ and $B$, and a mediator $M$ which may be finite or infinite-dimensional, so that the Hilbert space is the tensor product
$\mathcal{H}_A \otimes \mathcal{H}_B \otimes \mathcal{H}_M$.
The mediator interacts successively with systems $A$, $B$, $A$ again, and then $B$ again, resulting in the  unitary transformation
\begin{align}
 V_{ABM} = U_{BM}^{(2)} U_{AM}^{(2)} U_{BM}^{(1)}U_{AM}^{(1)}.
 \label{Eq:VABM}
 \end{align}
Note that this form is general, since any local evolution of $A$, $B$, or $M$ during the four-interaction sequence can be absorbed formally into one of the transformations $U_{A,B M}^{(1,2)}$. 

 We wish for $ V_{ABM} $ to induce an interaction between systems $A$ and $B$, while leaving them uncorrelated with the mediator at the end of the interaction. This requirement can be imposed at several levels, as discussed below.

At the weakest level, mediation may be catalytic only for a prescribed
initial mediator state $\ket{\Psi_i}_M$.  This \emph{state-dependent
mediation} is standard in quantum catalysis~\cite{RevModPhys.96.025005}, but
for other initial mediator states the final $AB:M$ state will generally not
factorize.  The stronger notion of \emph{state-independent mediation}
requires factorization for every initial state of $M$.

Independently, catalyticity may hold only at one tuned value of the
coupling, or throughout an interval around $g=0$.  Three SWAP gates provide
a simple state-independent but parameter-tuned example.  Such a construction
implements one transformation $U_{AB}$ but has no weak-$g$ regime that can be
interpreted as evolution under an effective Hamiltonian.  The precise
definitions, the SWAP construction, and the equivalence of the unrestricted
factorized forms under local controls are given in
Appendix~\ref{app:catalytic}.

\subsection{Fixed-Hamiltonian exact-return loops}

An effective interaction Hamiltonian requires more than one tuned gate:
the mediator must decouple throughout an interval around $g=0$.  We focus on
the strongest of the four notions just described, namely state-independent
mediation with exact mediator return,
\begin{equation}
V_{ABM} (g)   =U_{AB} (g)\otimes  \id_M \quad \forall g \in [-\delta , + \delta] ,
\label{Eq:VABM(g)}
\end{equation}
for some $\delta>0$.  We call this \emph{state-independent mediation of
effective interactions}.

\newcommand{\itsimei}{state-independent mediation of effective interactions}

We now restrict the local transformations to one-parameter groups
generated by fixed, interaction-only Hamiltonians,
\begin{align}
 V_{ABM} (g) &= e^{\ii g H_{BM}^{(2)}} e^{\ii g H_{AM}^{(2)}} 
 e^{\ii g H_{BM}^{(1)}} e^{\ii g H_{AM}^{(1)}},
 \label{Eq:VABMIntOnly-main}
\end{align}
where ``interaction only'' means that we exclude terms acting solely on $A$, $B$, or
$M$.   We
now summarize the constraints obtained by expanding 
Eqs. (\ref{Eq:VABM(g)}) and (\ref{Eq:VABMIntOnly-main})
in powers of $g$, i.e. $V_{ABM} (g)= \id + \sum_{n=1}^\infty g^n\,  V_{ABM}^{(n)}$. Their derivation
is given in Appendix~\ref{app:catalytic}.

\subsection{Perturbative conditions}

\paragraph*{Order $n=1$} As expected, at the lowest order $g$ no induced interaction between $A$ and $B$ is found. Nevertheless, the mediation condition~(\ref{Eq:VABM(g)}) imposes the constraints
\begin{align}
    H_{AM}^{(1)}=- H_{AM}^{(2)}:=H_{AM} \\
    H_{BM}^{(1)}=- H_{BM}^{(2)} := H_{BM}.
    \label{eq:H(1)}
\end{align}
To avoid the trivial situation $V_{ABM}(g)=\id_{ABM}$,
these Hamiltonians must have a nontrivial commutator $[H_{AM},H_{BM}]\neq 0$.

\paragraph*{Order $n=2$} At order $g^2$, we find that the state-independent mediation condition in Eq.~(\ref{Eq:VABM(g)}) further imposes
\begin{equation}
\label{eq:commutator-main}
   \ii \,[H_{AM},H_{BM}] =:   h_{AB}^{(2)} \otimes \id_{M},
\end{equation}
which can be seen to require an \emph{infinite-dimensional mediator}. A canonical realization is to take a mediator composed of a finite number $N$ of bosonic modes with the associated creation and annihilation operators satisfying $[a_i,a_j^\dag]=\delta_{ij}$, and local interactions that are linear in these operators
\begin{equation}
\begin{split}
H_{AM} &=  \sum_{i=1}^N (\hat A_i \otimes a_i +h.c.)\\
 H_{BM} &= \sum_{i=1}^N (\hat B_i \otimes a_i +h.c.),
\end{split}
 \label{Eq:HAMHBM-2}
\end{equation}
where $\hat A_i,\hat B_i$ are arbitrary traceless operators acting on $A$ and $B$, respectively.
This choice leads to the following effective interaction Hamiltonian between $A$ and $B$:
\begin{align}
h_{AB}^{(2)} 
= 
\ii \sum_{i=1}^N ( \hat A_i \hat B_i^\dag - \hat A_i^\dag \hat B_i),
\label{Eq:HhAB(2)-2}
  \quad \text{with} \\
  \label{Eq:HhAB(2)-3}
  V_{ABM} (g) 
    = \id - \ii  g^2  h_{AB}^{(2)} \otimes \id_{M} +O(g^3).
\end{align}

\paragraph*{Orders $n\geq 3$} At the next order we find
\begin{equation}\label{eq:third-order}
V_{ABM}^{(3)} = -\frac{1}{2} [H_{AM}+H_{BM},h_{AB}^{(2)}].
\end{equation}
In fact, using the Baker--Campbell--Hausdorff expansion, we see that the stronger condition 
\begin{equation}
 [H_{AM},h_{AB}^{(2)}] =  [H_{BM},h_{AB}^{(2)}] = 0
 \label{Eq:CatAllOrders1}
\end{equation}
is sufficient to guarantee Eq.~(\ref{Eq:VABM(g)}) at all orders $n\geq 3$. This gives the desired state-independent mediation of effective interactions valid at all $g$
\begin{align}
    V_{ABM} (g) 
    &= \exp \left( - \ii  g^2  h_{AB}^{(2)} \right) \otimes \id_{M}.
    \label{Eq:CatAllOrders2}
\end{align}

\subsection{Exact single-mode realization}
For a single-mode  ($N=1$) bosonic mediator in Eq.~(\ref{Eq:HAMHBM-2}), the sufficient higher order condition  defined by Eq.~(\ref{Eq:CatAllOrders1}) gives
\begin{align}
H_{AM} &= \sum_i (x_i a + \bar x_i a^\dag) \ketbra{i}_A
\quad x_i \in \mathbb{C} ,
 \label{Eq:HAM-1}\\
H_{BM}&= \sum_j (y_j a + \bar  y_j a^\dag) \ketbra{j}_B 
\quad y_j \in \mathbb{C} , \label{Eq:HBM-1}
\end{align}
where $\{\ket{i}_A\}$ and $\{\ket{j}_B\}$ are the bases of the Hilbert spaces diagonalizing $A$ and $B$, and the interaction-only hypothesis is satisfied by the choice
$\sum_i x_i=\sum_j y_j=0$. The mediated interaction then reads
\begin{align}
V_{ABM}(g)
&=\left(\sum_{i,j}
e^{g^2(x_i\bar y_j-\bar x_i y_j)}
\ketbra{ij}_{AB}\right) \otimes\id_M .
\label{Eq:N=1VAB2-B}
\end{align}

The local interactions $H_{AM}$ and $H_{BM}$ defined in Eqs. (\ref{Eq:HAM-1}) and (\ref{Eq:HBM-1}) generate phase-space displacements of the mediator mode controlled by the state of the system. In particular, the simple model in Eq. (\ref{Eq:Interactions}) is obtained by choosing all $x_i$ to be real and all $y_j$ to be imaginary, so that the respective displacements are in orthogonal phase-space directions.

These \emph{controlled-displacement interactions} provide a simple system--mediator hybridization mechanism which we analyse in Sec.~\ref{sec:dressed-states}.

\subsection{Weak mediation and the JC model}
\label{subsec-WeakInt}

The condition that a single mediator is catalytic to all orders, i.e.~Eq.~(\ref{Eq:CatAllOrders2}), is highly constraining. However, when the interaction strength is small
$g^2 \to 0$, one can 
require that the mediation condition hold only to leading order in $g^2$, giving rise to the interaction
Eqs. (\ref{Eq:HhAB(2)-2}, \ref{Eq:HhAB(2)-3}). 
Upon tracing out the mediator after the interaction loop, the next-order system--mediator coupling $V_{ABM}^{(3)}$ in Eq.~(\ref{eq:third-order}) can give an $O(g^3)$ correction to the interaction Hamiltonian (a Lamb-shift-like term) as well as a dissipation channel whose strength is $O(g^6)$. 

If the mediators are supplied every $\tau$, then defining 
$\mu = g^2/\tau$, we  obtain 
an effective interaction  given by 
\begin{equation}
H_{AB}^{\rm eff} = \mu \, h_{AB}^{(2)} = \mu \, \ii \sum_{i=1}^N ( \hat A_i \hat B_i^\dag - \hat A_i^\dag \hat B_i) \ . 
\end{equation}
This is known as the weak-coupling limit and has been studied in a variety of other contexts; see, e.g.~\cite{breuer2002theory,accardi1995stochastic,Attal2007}. However, for these models, we cannot in general  take $\tau$ too small. The reason is that when $\tau < 2T$ the interactions between the systems and successive mediators do not in general commute, which introduces terms of order $g^2$  and invalidates the catalytic conditions. Hence, for a fixed flight-time $T$ and since $g^2\ll 1$, the effective interaction strength is fundamentally limited to $\mu \ll 1/(2T)$ in the weak-coupling limit. 

Nevertheless, the weak-coupling case is interesting because there is  more freedom in designing the interactions, since
we do not require that $V_{ABM}^{(3)} $ and higher order terms vanish. 
As an illustration, we consider two-level atoms $A$ and $B$ coupled to a single-mode mediator through a Jaynes--Cummings (JC) interaction that exchanges an atomic excitation with a mediator boson (see also Ref.~\cite{KargEtAl2019}):
\begin{equation}\label{eq:JC-model}
H_{XM} = {\ii} \left(a \, \ketbra{1}{0}_X - a^\dag \, \ketbra{0}{1}_X \right)  , \quad X=A,B .
\end{equation}
Using Eq. (\ref{Eq:HhAB(2)-2}), one finds that this gives rise to the 
 effective  nonlocal Hamiltonian 
\begin{equation}\label{eq:HJCeff-B}
   H_{AB}^{\rm eff} = \frac{\mu }{2}(Y_A\otimes X_B - X_A\otimes Y_B).
\end{equation}
Note that there are many different local interactions $H_{AM}$ and $H_{BM}$ that can give rise to the same  effective  nonlocal Hamiltonian, i.e. Eq. (\ref{eq:JC-model}) is not the only model that gives rise to Eq. (\ref{eq:HJCeff-B}).

The JC model is exceptional because, even when there are a finite number of mediators present simultaneously between the systems $A$ and $B$, i.e. $\tau < 2T$, it remains catalytic at order $g^2$, see Sec.~\ref{sec:JC-dressed}.

\section{Hybridization and dressed states}
\label{sec:dressed-states}

\begin{figure}
 \includegraphics[width= \columnwidth]{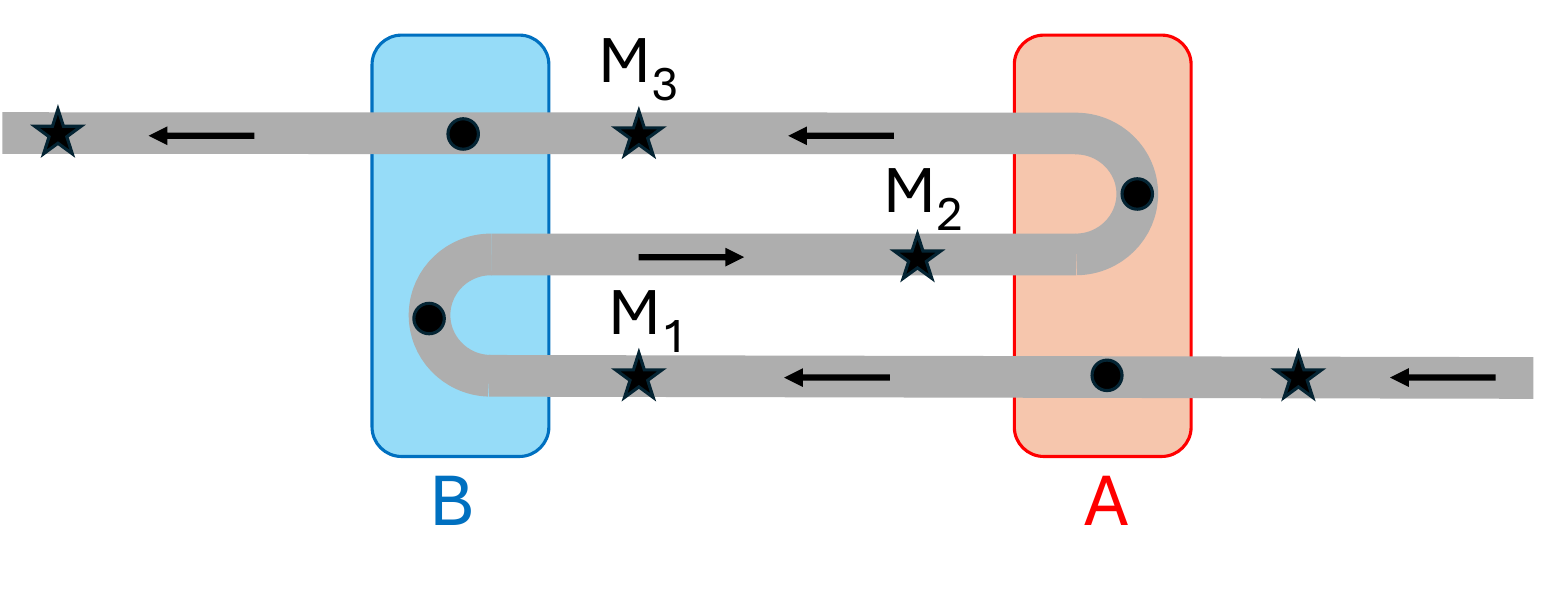}
\caption{
Spatial arrangement of systems $A$ (red) and $B$ (blue) and the mediators (stars), used to illustrate dressed states. The mediators arrive at a rate $1/\tau= 1/T$, ensuring that the mediator configuration repeats every $T$.
}
\label{Fig:strobo}
\end{figure}

We now show how the notion of a stationary dressed state, introduced in Sec.~\ref{sec:Hyb}, arises in our models. 
In order to do so, we consider a stroboscopic scenario intermediate between the two situations depicted in Fig. \ref{Fig:ABM-FF}. Namely we suppose that a new mediator is supplied every $\tau$. This ensures that we have discrete time-translation symmetry (the mediator configuration reproduces exactly every $\tau$). We chose $\tau$ such that at times $t_n:=\tau n$ there is exactly one mediator in each ``section of the waveguide'', i.e.~one mediator that has interacted once, one twice, and one three times with the systems. The waveguide picture is illustrated in Fig.~\ref{Fig:strobo} for the specific choice $\tau=T$. 

We are looking for the stationary states 
\begin{equation}
\ket{\Psi(t_n)}_{ABM}  \quad t_n = n\,\tau
\end{equation}
which are invariant, up to an $n$-dependent phase which will correspond to the interaction energy, under the action of $V_{ABM}(g)$, now understood to act on all the mediators present in the system.
In the case of perfect mediation, the mediator modes that have already completed the four interactions return to their initial state, and can be ignored. At times $t_n$ the systems are thus only correlated with the three mediator modes $M=M_1M_2M_3$ (identified in Fig.~\ref{Fig:strobo}), that have already entered, but not yet completed the interaction loop. 

We now analyse this stroboscopic scenario for the controlled-displacements model of Eq.~(\ref{Eq:Interactions}), and the  weak-coupling limit of the Jaynes-Cummings model of Eq.~(\ref{eq:JC-model}).

\subsection{Dressed states in the qubus model}

When the local interactions $U_{AM}^{(1,2)}$ and $U_{BM}^{(1,2)}$ are the controlled-displacements of Eq.~(\ref{Eq:Interactions}), the stationary states $\ket{\Psi}_{ABM}$ are easily computed. 
 Indeed, let $A=\sum_a\lambda_a\ketbra{a}_A$ and
$B=\sum_b\lambda_b'\ketbra{b}_B$ be the spectral decompositions of the coupling operators. These allow us to write 
\begin{align}
\ket{\Psi(t_n)}_{ABM} &= \sum_{a,b} c_{ab}   \ket{a,b}_{AB}  \ket{\phi_{ab}(t_n)}_{M}\ ,
\label{eq:strob3}
\end{align}
where the amplitudes $c_{a,b}$ are stationary since the interaction is controlled on the values $a$ and $b$.
Initially, all the mediators are taken to be in the same state  $\ket{\phi_0}$. Once a mediator has undergone the four interactions, it goes back to its initial state, having acquired the (geometric) phase $e^{i g^2 \lambda_a \lambda'_b}$, see Eq. (\ref{eq:VABM}).
Therefore, we have
\begin{align}
& \ket{\phi_{ab}(t_n)}_{M} = e^{i n g^2 \lambda_a \lambda'_b }   \left(  
       e^{ \ii g^2 \lambda_a \lambda_b'}\mathcal{D}\Big(-g\frac{\lambda_b'}{\sqrt{2}}\Big)\ket{\phi_0} \right)_{M_3} \nonumber\\
       &\left( 
 e^{\ii \frac{g^2}{2}\lambda_a \lambda_b'}\,
 \mathcal{D}\Big(g\frac{\ii \lambda_a-\lambda_b'}{\sqrt{2}}\Big)  \ket{\phi_0}\right)_{M_2}\!\!\left(    \mathcal{D}\Big(g\frac{\ii \lambda_a }{\sqrt{2}}\Big)\ket{\phi_0}   \right)_{M_1} 
,
 \label{eq:strob2}
\end{align}
where $\mathcal D(\alpha)=e^{\ii \sqrt 2({\rm Im}[\alpha] x -{\rm Re}[\alpha]p)}$ is the displacement operator, and
we have used the fact that the product of two displacements $\mathcal D(\beta) \mathcal D(\alpha)=e^{\ii \, \operatorname{Im}[\beta \alpha^*]}\mathcal{D} (\alpha+\beta)$ is a displacement times a phase. 
In the microscopic model the eigenstates $\ket{a,b}_{AB}$ of the effective Hamiltonian thus correspond to the stationary dressed states $\ket{a,b}_{AB}\ket{\phi_{ab}(t_n)}_{M}$ involving the three displaced mediators $M_1$, $M_{2}$, and $M_{3}$, whose displacements are controlled by the values $\lambda_a$ and $\lambda_b'$.

\subsection{Purity of the reduced system state}

Although the eigenstates $ \ket{a,b}_{AB}  \ket{\phi_{ab}(t_n)}_{M}$ 
remain factorized, their superpositions
in Eq. (\ref{eq:strob3}) 
are generally entangled across the $AB:M$ cut. The amount of entanglement depends on $g$ and the initial state $\ket{\phi_0}_M$ of the mediator.  The presence of this entanglement  stems from the same noncommutativity of the controlled displacements that gives rise to the interaction energy $g^2 \lambda_a \lambda'_b$ in Eq. (\ref{eq:strob2}).

Let us assume that the  mediators are supplied in the vacuum state 
\begin{equation}
\ket{\phi_0} = \ket{0}.
\label{Eq:vac}
\end{equation}
In this case we can compute the overlaps
\begin{equation}\label{eq:branch-overlap}
\left|\braket{\phi_{ab}(t_n)}{\phi_{a'b'}(t_n)}_M\right|
=
\exp\!\left[
-\frac{g^2}{2}\,\Delta_{ab,a'b'}
\right]\,
\end{equation}
with $\Delta_{ab,a'b'}
:=
(\lambda_a-\lambda_{a'})^2+(\lambda_b'-\lambda'_{b'})^2$.
Tracing out the mediator modes   gives
\begin{align}\label{eq:purity-exact}
\Tr(\rho_{AB}^2) &=
\sum_{a,b,a',b'} |c_{ab}\, c_{a'b'}|^2\,e^{- g^2\Delta_{ab,a'b'}}
\\
\label{eq:purity-weak}
&\approx
1-2 g^2\left(\mathrm{Var}_\Psi(A)+\mathrm{Var}_\Psi(B)\right) ,
\end{align}
where in the second line we have taken the limit $g^2 \ll 1$.
Thus the reduced state becomes mixed whenever its superposition contains branches with distinguishable $A$ or $B$ eigenvalues, equivalently whenever $\mathrm{Var}_\Psi(A)+\mathrm{Var}_\Psi(B)>0$.
The heuristic argument leading to Eq.~(\ref{Eq:purity}) anticipated an impurity, but the microscopic model gives a larger effect: it scales as $O(g^2)=O(\mu T)$ rather than $O(\mu^2 T^2)$.

\subsection{Dressed states in the JC model}

\label{sec:JC-dressed}

As mentioned above, second-order catalyticity for one isolated mediator is not, by itself, sufficient to ensure catalyticity when several mediators are simultaneously present in the loop. 

However, the local Jaynes--Cummings model of Eq.~(\ref{eq:JC-model}) is exceptional for the particular stroboscopic implementation considered here.  When vacuum mediators arrive every $T$, as in Fig.~\ref{Fig:strobo}, the  JC model admits four states that are stationary up to the order $g^2$. The full derivation is given in Appendix~\ref{app:JC-dressed}, here we give the general picture.

Since the JC model preserves the number of excitations, the state with no excitations
\begin{equation}
\ket{\Psi_{00}}_{ABM}=\ket{0,0}_{AB} \prod_i \ket{0}_{M_i} .
\end{equation}
is  exactly stationary. The  three additional \emph{quasi-stationary} dressed states are
\begin{align}
     \ket{\Psi_{\pm}(t_n)}_{AB M}= &
    e^{\pm \ii n g^2} \Bigl(
\ket{\psi_\pm}_{AB}\ket{0}_{ M}
\nonumber\\
&
+g \ket{0,0}_{AB}\ket{\phi_\pm (n)}_{ M}+ O(g^2)
\Bigr)     
\nonumber
     \\
     \ket{\Psi_{11}(t_n)}_{AB M} =& \ket{1,1}_{AB}\ket{0}_{ M}
     +g \Bigl( \ket{1,0}_{AB}\ket{\phi_{10}(n)}_{ M} 
\nonumber\\     
     &+ \ket{0,1}_{AB}\ket{\phi_{01}(n)}_{ M} \Bigr)   + O(g^2), \nonumber
     \label{eq:Psi11}
\end{align}
where $\ket{\psi_\pm} = \frac{1}{\sqrt 2} (\ket{10}\pm \ii \ket{01})$ and $\ket{\phi_\bullet(n)}_{ M}$ are different states of the mediators $M=M_1 M_{2} M_{3}$ with exactly one excitation. The four dressed states $\ket{\Psi_{00}}_{ABM}, \ket{\Psi_{\pm}}_{ABM}$ and $\ket{\Psi_{11}}_{ABM}$  provide a microscopic description for the eigenstates  $\ket{00}_{AB}, \ket{\psi_\pm}_{AB}$ and $\ket{11}_{AB}$  of the nonlocal Hamiltonian  Eq.~(\ref{eq:HJCeff-B}).
Note that, contrary to the simple model described above, in the Jaynes-Cummings model the dressed states can themselves be entangled across the $AB:M$ cut.

\section{Continuous field limit}
\label{sec:continuous-field}

The stroboscopic scenario makes the hybridization mechanism transparent but does not provide a continuous-time-translation-invariant interaction. To obtain such a description we now consider a scenario with a continuous flux of mediating bosons, depicted in Fig.~\ref{Fig:ABM-FF}(b), focusing on our simple controlled-displacement interactions. 

The simplest route to the continuous-field limit consists of taking the stroboscopic scenario to the limit with mediator modes arriving at a diverging rate $1/\tau\to\infty$~\cite{Attal2003}. Again, it is useful to view the mediators as traveling in a waveguide, as shown in Fig.~\ref{Fig:Mediators-Schrodinger}, so that they can be localized in both space and time.
The flight time $T$ and distance $L$ between $A$ and $B$ are divided into $n$ intervals $\tau= T/n$ and $l=L/n$. We introduce a train of bosonic modes, described by ladder operators $a_i,a_i^\dag$, with $n$ modes in every segment of length $L$ (and hence $3n$ modes in the three-segment interaction region). At each time step a fresh bosonic mode enters the loop, all modes shift forward by one site, and the modes at the four interaction points undergo the controlled displacements of Eq.~(\ref{Eq:Interactions}), controlled on  the eigenvalues of the coupling operators $A=\sum_a\lambda_a\ketbra{a}_A$ and $B=\sum_b\lambda_b'\ketbra{b}_B$. When the systems are prepared in the respective eigenstates $\ket{a}_A$ and $\ket{b}_B$, the modes inside the waveguide have \emph{stationary} states determined by the corresponding eigenvalues $(\lambda_a,\lambda_b')$. As already mentioned (see Appendix~\ref{app:ContField}), each mediator completing the interaction loop returns to its initial state while acquiring a geometric phase $e^{\ii g^2\lambda_a\lambda_b'}$. 
Given the rate $\nicefrac{1}{\tau}$ at which mediators flow, this phase corresponds to the effective energy
\begin{equation}
E_{ab}
=-\frac{g^2}{\tau}\lambda_a\lambda_b'
=-\frac{n g^2}{T} \lambda_a\lambda_b'
\end{equation}
associated with the system state $\ket{a,b}_{AB}$.

\begin{figure}
 \includegraphics[width=\columnwidth]{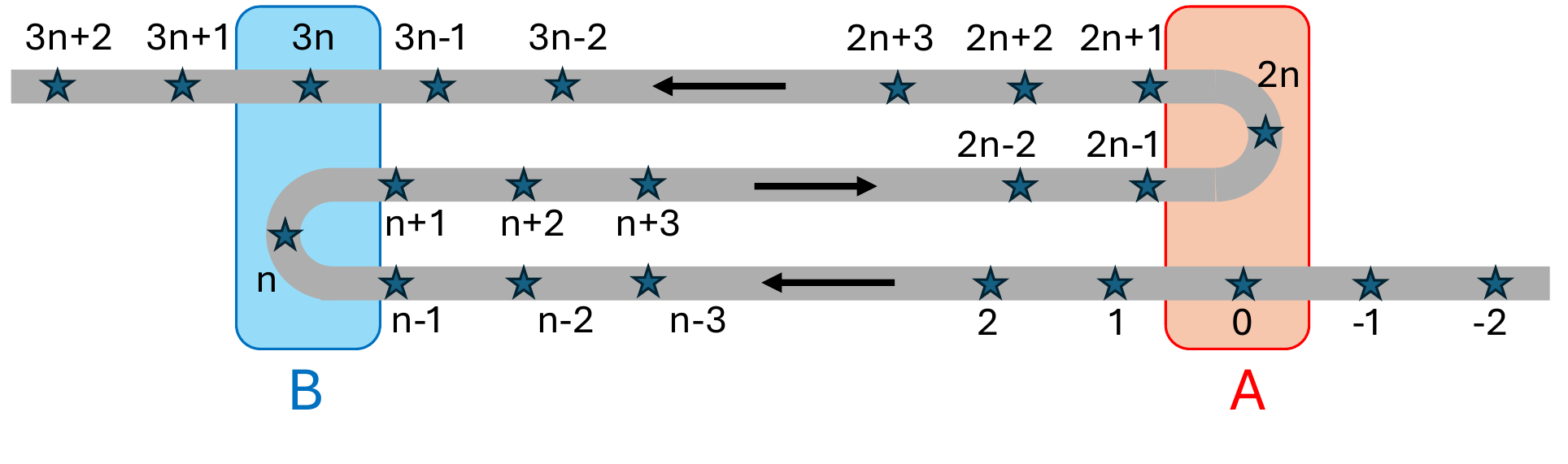}
\caption{
Spatial arrangement of systems $A$ (red) and $B$ (blue) and the mediator modes (stars). There are $n$ modes between $A$ and $B$. In one elementary time step the modes hop forward by one position, and those at positions $0$, $n$, $2n$, and $3n$ interact with $A$ or $B$.
}
\label{Fig:Mediators-Schrodinger}
\end{figure}

The continuum limit is reached by sending $n\to \infty$  and $g\to 0$ while keeping a fixed effective interaction strength $\mu:=\nicefrac{n g^2}{T} = \frac{g^2}{\tau}$. The energies of the eigenstates are $E_{ab}=-\mu\,\lambda_a \lambda_b'$, corresponding to the effective Hamiltonian 
\begin{equation}
H_{AB}^{\rm eff}=-\mu\,A  B.
\end{equation}
In the limit $n\to \infty$, we have 
 $\dd t = \frac{T}{n}$ and $\dd x = \frac{L}{n}$. The creation and annihilation field operators, obeying $[  a(x),  a^\dag(x')] =\delta(x-x')$, are given by $  a(x)= \frac{a_i}{\sqrt{\dd x}}$ with $x=\frac{i}{n} L$ and define a chiral bosonic field that propagates through the waveguide at the velocity $v=\nicefrac{L}{T}$. 
 
Specializing to the case where 
 the field is initially in the vacuum state $\ket{0}_{M}$, the stationary states of the whole system take the form
\begin{align}
&\ket{\Psi_{ab}}_{AB  M}
=
\ket{a,b}_{AB}\ket{\mathfrak F_{ab}}_{ M} \quad \text{with}
\\
&\ket{\mathfrak F_{ab}}_{ M}
\propto
\exp\!\left(
\kappa \int_0^{3L}  \! \! \! \!\dd x\, \zeta_{ab}\left[\nicefrac{x}{L}\right]\,   a^\dag(x)
\right)\ket{0}_{ M},
\end{align}
where $\kappa := \sqrt{\frac{\mu}{2 v }}$ and the coherent-field profile inside the waveguide is
\begin{equation}\label{eq:zeta-profile}
\zeta_{ab}(r)=
\begin{cases}
i\,\lambda_a & 0\le r<1,\\
i\,\lambda_a-\lambda_b' & 1\le r<2,\\
-\lambda_b' & 2\le r<3,
\end{cases}
\end{equation}
while the field outside remains in the vacuum. 

The continuous field construction extends naturally to any single-boson model of the form~(\ref{Eq:N=1VAB2-B}), because there is a basis in which the system states are unchanged by the system--mediator interactions. 

\section{Decoherence under local operations}
\label{sec:localDecoh}

\begin{figure*}[t]
\centering
\includegraphics[width=1.5\columnwidth]{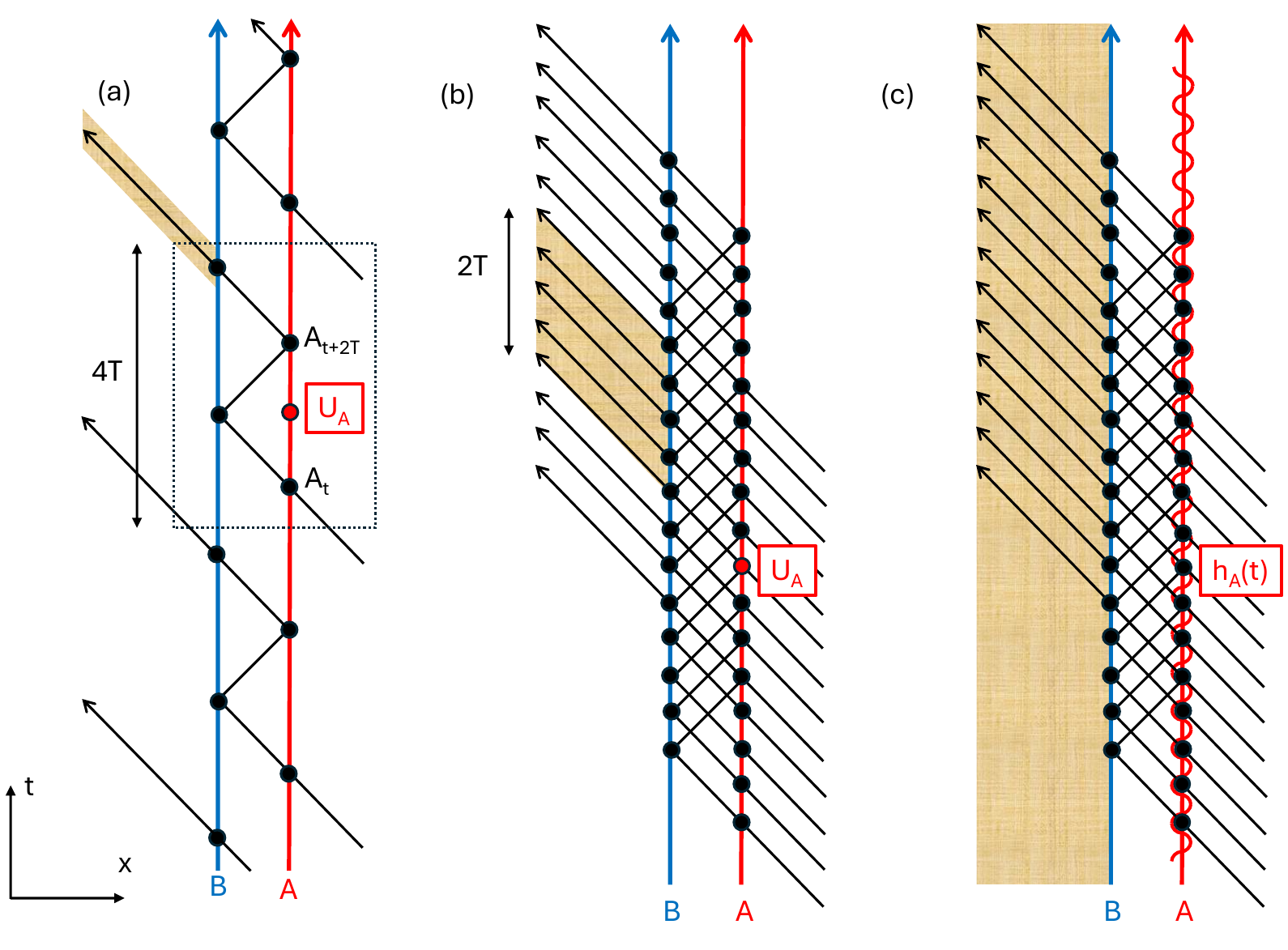}
\caption{Local control resolves the mediator dynamics. (a) In a discrete loop, arbitrary evolution of Alice's coupling observable between the two encounters is captured by the mismatch $\Delta A=A_{t+2T}-A_t$; sudden gates and slowly varying control are two limits of the same expression. (b) For a continuous mediator flux, an instantaneous gate is exactly tractable and emits a disturbed pulse of duration $2T$. (c) The most general setting combines a continuous mediator field with continuous local control $h_A(t)$.  In this case, a continuous flux of disturbed mediators is emitted. Its detailed analysis is left open. }
\label{Fig:ABM-Flux}
\end{figure*}

The effective interaction relies on destructive interference: the mediator's later encounters undo the displacements created by the earlier ones. Fast local operations on the systems change the coupling observable between these encounters and can prevent the loop from closing. We now discuss this effect.

We will assume that only Alice's system is subject to local manipulation. This can include sudden gates $U_A$ applied at specific times or time-dependent driving with a local Hamiltonian $h_A(t)$. The combination of their effects can be collected into a time-dependent unitary $u_A(t)$ describing the local evolution of Alice's system between time 0 and $t$ (in the absence of interaction with the mediator). Accordingly, its evolution between the times $t_1$ and $t_2\geq t_1$ is given by $u_A(t_2)u_A^\dag(t_1)$.

\subsection{Discrete mediator loops}
\label{Subsec:SingleMediator}

Consider a single interaction loop of Sec.~\ref{sec:single-boson} with the mediator entering at time $t$. In the Schr\"odinger picture let us express the global propagator at time $t+3T$ after the mediator completes its loop. Taking the local evolution of $A$ into account we find
\begin{align}
    V_{ABM}^{(t+3T)} = &e^{-\ii g B p} u_A(t+3T)u_A^\dag(t+2T) e^{-\ii g A x} \\
    &u_A(t+2T)u_A^\dag(t+T) e^{\ii g B p} \\
    &u_A(t+T)u_A^\dag(t) e^{\ii g A x} u_A(t).
\end{align}
The expression is cumbersome, and it is more convenient to represent the propagator in the interaction picture $\widetilde V_{ABM}^{(t+3T)} := u_A^\dag(t+3T) V_{ABM}^{(t+3T)}$ with Alice's coupling observable carrying an explicit time dependence. Introducing such dependence allows us to write
\begin{align}
&A_t:=u_A^\dagger(t)A\,u_A(t) \\
&\widetilde V_{ABM}^{(t+3T)}
=
e^{-igBp}\,e^{-igA_{t+2T}x}\,
e^{\ii gBp}\,e^{\ii gA_tx}.
\label{eq:controlled-loop}
\end{align}

One recognizes here that the crucial element is the change in Alice's observable 
\begin{equation}
    \Delta A_t:=A_{t+2T}-A_t 
\end{equation}
between the two encounters with the mediator. When $\Delta A_t=0$, Eq.~(\ref{eq:controlled-loop}) reduces to the closed loop of Sec.~\ref{sec:single-boson}. Otherwise an uncanceled displacement remains and the outgoing mediator carries information about $A$. When $\Delta A_t \neq 0$ the mediator does not generally close its displacement loop, it remains correlated with Alice's system and induces decoherence of the latter once it is traced out. 

Let us now discuss this effect in detail, assuming a small local interaction strength $g$ and the mediator starting in the vacuum state. 
The dependence on the initial mediator state is discussed in Appendix~\ref{app:pert}. In the interaction picture, the evolution of the system after one interaction loop is given by the channel $\mathcal{V}[\rho_{AB}] := \tr_M \widetilde V_{ABM}^{(t+3T)}( \rho_{AB}\otimes \ketbra{0}{0}_M )\widetilde V_{ABM}^{(t+3T)\dag}$. In Appendix \ref{app:pert} we show that it reads
\begin{align}
\mathcal{V}[\rho_{AB}]=  
\rho_{AB}
&+\ii g^2[A_{t+2T}B-H_A^{\rm L},\rho_{AB}]
\nonumber\\
&+\frac{g^2}{2}\mathcal D_{\Delta A_t}(\rho_{AB})
+O(g^3),
\label{eq:controlled-loop-map}
\end{align}
where
\begin{align}
H_A^{\rm L}&:=\frac{\ii}{4}[A_{t+2T},A_t],\\
\mathcal D_L(\rho)&:=L\rho L^\dag-\frac12\{L^\dag L,\rho\}.
\end{align}
The first term contains the effective interaction and a Lamb-shift-like correction. The second is the residual decoherence caused by the mismatch. For other initial mediator states the same mechanism remains, although its strength depends on the relevant quadrature uncertainty.

Let us now further assume that $u_A(t)$ stems from a local Hamiltonian $h_A(t)$ that is slowly varying, i.e.,~$h_A(t)\approx h_A(t+2T)$, and also slow on the timescale of the mediator flight time, i.e.,~$u_A(t+2T)u_A^\dag(t) \approx \id -\ii\, 2 T\, h_A(t)$. With $h_A^{\rm I}(t):=u_A^\dag(t)h_A(t)u_A(t)$, this allows us to write
\begin{align}
    \Delta A_t &\approx  2T\,  \ii [h_A^{\rm I}(t), A_t] = 2 T \,  \dot A_t \ .
\end{align}
Such slow control thus produces a dissipative correction governed by $T\dot A_t$.

Keeping the slow-control picture, let new mediators arrive at time intervals $\tau>2T$, so that their interaction loops are nonoverlapping. This recovers an effective reduced dynamics of systems $A$ and $B$, which is now generated by a Liouvillian rather than a Hamiltonian. Indeed, as $g^2$ is small and each channel $\mathcal{V}$ has an
effective short duration $\tau$, letting $\dot \rho_{AB} = \frac{\mathcal{V}[\rho_{AB}]-\rho_{AB}}{\tau}$ we can rewrite Eq.~(\ref{eq:controlled-loop-map}) in the form of a master equation
\begin{align}
    \dot \rho_{AB} = \ii \mu [A_{t+2T} B - H_A^L,\rho_{AB}] +2 \mu \, \mathcal{D}_{T \dot A_t}(\rho_{AB}),
\end{align}
with the usual convention $\mu = \frac{g^2}{\tau}$. The dissipation rate is thus controlled by $r \sim \mu T^2  \|\dot A_t\|^2$. In turn, to implement a desired local gate, the local Hamiltonian $h_A(t)$ has to be applied for a total duration $t_{\rm gate} \sim 1/ \| \dot A_t \|$. Thus, we find that during this process the accumulated strength of decoherence is roughly $t_{\rm gate}\, r \sim \mu T^2 \|\dot A_t\|\sim \frac{\mu T^2}{t_{\rm gate}}$, it increases with the mediator flight time and the effective interaction strength, and decreases with the gate duration. In this context, $t_{\rm gate} \gg \mu T^2$ is the condition under which the implementation of local gates does not induce decoherence.

\subsection{An instantaneous gate in the continuous-field model}

Switching to the continuous-field model of Sec.~\ref{sec:continuous-field}, we restrict our analysis to a single sudden gate $U_A=\sum_{a,a'}U_{a'a}\ketbra{a'}{a}_A$ applied at time $t=0$. Consider its effect on a stationary branch
$\ket{a,b}_{AB}\ket{\mathfrak F_{ab}}_{\bm M}$.  Immediately after the gate this state is transformed to
\begin{equation}
U_A\ket{a,b}_{AB}\ket{\mathfrak F_{ab}}_{M}
=
\sum_{a'}U_{a'a}\ket{a',b}_{AB}\ket{\mathfrak F_{ab}}_{\bm M}.
\label{Eq:perturbed}
\end{equation}
On the right-hand side, components with $a'\neq a$ are not stationary because the field dressing still corresponds to the pre-gate value $\lambda_a$. For a duration $2T$, the later displacement controlled by $\lambda_{a'}$ fails to cancel the earlier displacement controlled by $\lambda_a$, resulting in a coherent mediator field leaving the waveguide with an amplitude controlled by the displacement mismatch. For all $t\geq3T$, the field inside the waveguide has relaxed to the new dressing, while some information about the displacement mismatch encoded in a coherent pulse has left the waveguide. Up to a branch-dependent dynamical phase, each component $\ket{a',b}_{AB}\ket{\mathfrak F_{ab}}_{\bm M}$ has evolved to
\begin{align}
\ket{a',b}_{AB}
\ket{\mathfrak F_{a'b}}_{M_i}
\ket{\alpha_{aa'}(t)}_{ M_o},
\end{align}
where $ M_i$ and $ M_o$ denote the field inside and outside the waveguide, respectively. The outgoing pulse is
\begin{equation}
\ket{\alpha_{aa'}(t)}_{M_o}
\propto
\exp\!\left(
\ii\kappa\!\int_{vt}^{vt+2L}\!\dd x\,
\Delta\lambda_{aa'}\,a^\dagger(x)
\right)\ket{0}_{M_o},
\end{equation}
where
$\Delta\lambda_{aa'}:=\lambda_a-\lambda_{a'}$. The escaping field therefore records the change $\Delta \lambda_{aa'}$ in Alice's coupling eigenvalue across all transitions driven by $U_A$. The overlap of two outgoing pulses is
\begin{equation}
\braket{\alpha_{\bar a\bar a'}}{\alpha_{aa'}}
=
\exp\!\left[
-\frac{\mu T}{2}
\bigl(
\Delta\lambda_{aa'}
-\Delta\lambda_{\bar a\bar a'}
\bigr)^2
\right],
\end{equation}
and the decoherence strength is controlled by $\mu T/2$. Appendix~\ref{App:ContLocalGate} gives the full time-dependent field calculation.

\subsection{Continuous control in the continuous field}

The preceding result still idealizes Alice's operation as instantaneous. The most general case, represented in Fig.~\ref{Fig:ABM-Flux}(c), has a continuous local Hamiltonian $h_A(t)$ acting while a continuous field simultaneously propagates through the waveguide. In the interaction picture, the coupling observable becomes time-dependent, $A_t=u_A^\dagger(t)Au_A(t)$, and observables evaluated at different times need not commute. This makes a closed analytical formula challenging to obtain. Nonetheless, 
for sufficiently slow control, the decoherence strength is still controlled by $A_{t+2T}-A_t\simeq2T \dot A_t$.

\subsection{Decoherence in the JC model}

The decoherence mechanism in the Jaynes--Cummings model of Secs.~\ref{subsec-WeakInt} and~\ref{sec:JC-dressed} differs from the dephasing mechanism of the controlled-displacement model.  The JC dressed states have a fixed total excitation number, and their stationarity relies on destructive interference that prevents the field from carrying an excitation out of the interaction region. Therefore, even in the absence of local control, any initial state that is not stationary will decay to a mixture of stationary states with fewer excitations. A fast local gate acts only on the atomic component and generally does not transform between dressed states. 
The resulting decoherence is therefore a relaxation to lower-excitation dressed states rather than dephasing in the dressed-state basis.

\section{Back to the paradoxes}
\label{sec:back-paradoxes}

The two paradoxes of the introduction have the same origin: $H_{AB}^{\rm eff}$ is a coarse-grained generator obtained after eliminating the mediator, not a microscopic interaction acting instantaneously on two separated systems. In the microscopic description every coupling is local, and any disturbance must propagate through the mediator before it can influence the remote system.

Consider first the signaling paradox. A local operation $U_A$ changes only Alice's system and the mediator degrees of freedom in her causal neighborhood. If $T=L/v$ is the one-way travel time from Alice to Bob, locality implies
\begin{equation}
\rho_B(t|U_A)=\rho_B(t|\id_A),\qquad 0<t<T.
\end{equation}
Alice's change can affect Bob only after the mediator has carried it to him. The immediate change of $\rho_B$ predicted by applying $e^{-\ii tH_{AB}^{\rm eff}}$ for arbitrarily small $t$ therefore lies outside the time scale on which the effective Hamiltonian is valid. In our mediator models the correction is explicit: after a local gate, the modified dressing propagates through the waveguide as a field disturbance carrying information about Alice's transition.

The apparent energy paradox is resolved by the same separation of time scales. In a microscopic local description, the work performed by a fast gate is the energy supplied by its local control apparatus to Alice's system and to nearby mediator degrees of freedom. It cannot depend instantaneously on a remote choice by Bob. What fails during a fast operation is the identification of this local work with the change of the coarse-grained quantity $\langle H_{AB}^{\rm eff}\rangle$, because the mediator contributions to the interaction energy have been eliminated from that quantity.

In the two-qubit example, Alice's $X_A$ gate and Bob's $X_B$ gate each disturb their local system--mediator configuration. These disturbances propagate and may later interfere or redistribute energy, but neither control can respond instantaneously to the distant choice. The fact that the final effective interaction energy is unchanged when both gates are applied is therefore a statement about the coarse-grained long-time $AB$ description, not a statement that Alice's microscopic work vanished. An explicit energy accounting would additionally require the control apparatus to be included in the microscopic Hamiltonian.

In summary, the effective interaction picture is valid when  local control changes the coupling observables little over the mediator flight time, $T\|\dot A\|\ll\|A\|$. When this condition is not obeyed decoherence and radiation reveal the presence of a causal microscopic interaction mechanism. Apparent nonlocal effects are replaced by causal processes.

\section{Outlook}
\label{sec:outlook}

We end with several open questions.

\emph{Unavoidable correlations and decoherence.} Does finite propagation time impose a model-independent lower bound on correlations and entanglement between the systems and mediator for a prescribed effective interaction? Similarly, can the decoherence produced by local operations be bounded in terms of the mediator travel time, effective interaction strength, and control rate? What are the microscopic models that minimize these effects?

\emph{Distinguishing microscopic models.} Different mediators and local couplings can generate the same $H_{AB}^{\rm eff}$ while producing different dressings, Lamb shifts, and decoherence under fast control. To what extent can measurements on $A$ and $B$ distinguish these microscopic ``UV completions'' of the same effective interaction?

\emph{Quantum-field models.} Can quantum-information-inspired microscopic models provide simple but physically faithful descriptions of interactions in quantum field theory?

\emph{Many-body extensions.} It would be natural to replace each interaction bond of a lattice model by a local propagating mediator channel. Such constructions could clarify how microscopic propagation, mediator correlations, and information leakage modify effective many-body dynamics and its relation to Lieb--Robinson bounds.

\begin{acknowledgments}
We thank Sandu Popescu, Eduardo Martinez and Alejandro Pozas-Kerstjens for insightful discussions.
\\

\emph{Author contributions:} All authors contributed equally; authors are listed alphabetically.
\end{acknowledgments}

\bibliographystyle{apsrev4-2}
\IfFileExists{refs_2.bib}{\bibliography{refs_2}}{\bibliography{refs_2}}

\clearpage

\appendix

\onecolumngrid

\section{Conditions for a catalytic mediator -- definitions and proofs} \label{app:catalytic}

We first give  precise definitions for the different kinds of catalytic mediation that were summarized in
Sec.~\ref{catalytic-main}. 

\subsection{State-dependent mediation.}

 At the weakest level, mediation may be catalytic
only when the mediator is prepared in a prescribed initial state
$\ket{\Psi_i}_M$.  Formally, state-dependent mediation means
\begin{equation}
    V_{ABM}\ket{\Psi_i}_M=U_{AB}\ket{\Psi_f}_M,
    \label{eq:state-dependent}
\end{equation}
where the equality is understood as an equality of maps on arbitrary
input states of $AB$.  Note that a final local unitary on $M$ can always return
$\ket{\Psi_f}_M$ to $\ket{\Psi_i}_M$.  For other initial mediator states, the final $AB:M$ state will generally not factorize.

\subsubsection{State-independent mediation}
\label{Subsub:State-independmed}

State-independent mediation instead requires factorization for every
initial mediator state, and hence
\begin{align}
V_{ABM}=U_{AB}\otimes U_M.
\label{eq:catalytic-general}
\end{align}
At the level of arbitrary local unitaries, this condition may be reduced
without loss of generality to
\begin{align}
V_{ABM}=U'_{AB}\otimes\id_M,
\label{eq:catalytic-general-2}
\end{align}
with $U'_{AB}$ interaction only, meaning that $\log U'_{AB}$ is
orthogonal to all local terms $A\otimes\id_B$ and $\id_A\otimes B$ in the
Hilbert--Schmidt product.  Indeed, a general bipartite unitary can be written \cite{AlekseevskyNikonorov2009}
\begin{equation}
U_{AB}=U_AU_BU'_{AB}U'_AU'_B,
\end{equation}
where $U_A,U_B,U'_A,U'_B$ are local.  These local factors, together
with $U_M$, can be absorbed into the four unrestricted $AM$ and $BM$
unitaries.

\subsubsection{Parameter-tuned mediation}

Independently of the initial-state requirement, catalyticity may hold
only at one tuned interaction strength.  A state-independent example is the
three-SWAP construction (for identical $A$, $B$, and $M$),
\begin{align}
V_{ABM}
&={\rm SWAP}_{AM}{\rm SWAP}_{BM}{\rm SWAP}_{AM}\nonumber\\
&={\rm SWAP}_{AB}\otimes\id_M.
\end{align}
It implements one gate but has no weak-$g$ regime that can be
interpreted as evolution under an effective Hamiltonian.  

\subsubsection{Effective Hamiltonians}

To obtain an effective Hamiltonian, 
we suppose that the interactions can be written as
\begin{align}
 V_{ABM} (g) &= e^{\ii g H_{BM}^{(2)}} e^{\ii g H_{AM}^{(2)}} 
 e^{\ii g H_{BM}^{(1)}} e^{\ii g H_{AM}^{(1)}} 
 \label{Eq:VABMIntOnly}
\end{align}
and that
mediator decoupling holds throughout an interval around
$g=0$.  

State-dependent interval mediation requires
\begin{equation}
V_{ABM}(g)\ket{\Psi_i}_M
=U_{AB}(g)\ket{\Psi_f(g)}_M
\quad\forall g\in[-\delta,+\delta],
\label{Eq:VABM(g)-initial}
\end{equation}
whereas the state-independent exact-return condition is
\begin{equation}
V_{ABM} (g)   =U_{AB} (g)\otimes  \id_M \quad \forall g \in [-\delta , + \delta] .
\label{Eq:VABM(g)-2}
\end{equation}

In what follows, we consider the state-independent exact-return condition,
and require that the system-mediator Hamiltonians $H_{XM}^{(i)}$, $X=A,B$, $i=1,2$, are interaction only.

We note that the absorption argument leading from Eq.~(\ref{eq:catalytic-general}) to
Eq.~(\ref{eq:catalytic-general-2}) concerns unrestricted local unitaries.  In the
fixed-Hamiltonian setting of Eq.~(\ref{Eq:VABMIntOnly}), absorbing a
$g$-dependent $U_M(g)$, $U_A(g)$, or $U_B(g)$  into one local Hamiltonian generally produces a
$g$-dependent system or mediator-only term.  It therefore need not
preserve the interaction-only Hamiltonians that we assume. Thus Eq.~(\ref{Eq:VABM(g)-2}) defines the
exact-return subclass used below.

\subsection{Order-by-order conditions}

We consider instances of \itsimei
of the form
\begin{align}
 V_{ABM} (g) &= e^{\ii g H_{BM}^{(2)}} e^{\ii g H_{AM}^{(2)}} 
 e^{\ii g H_{BM}^{(1)}} e^{\ii g H_{AM}^{(1)}} 
\end{align}
  We  expand $ V_{ABM} (g)$ in powers of $g$ 
\begin{align}\label{eq:V-exp}
 V_{ABM} (g)= \id + \sum_{n=1}^\infty g^n\,  V_{ABM}^{(n)} .
\end{align}

We impose the catalytic condition order by order. At order $k$, we  write
\begin{align}
 V_{ABM} (g) &=e^{-\ii H_{AB}^{(k)}(g)}  \otimes \id_M + O(g^{k+1})
 \label{Eq:VABMIntOnly-BB}
\end{align}
where
\begin{align}
H_{AB}^{(k)}(g) & =\sum_{n = 1 }^k g^n h_{AB}^{(n)}   
\label{Eq:hABexpansion}
\end{align}
with
 \begin{align}
 h_{AB}^{(n)} \text{ interaction only } \forall n.
\label{Eq:Interactiononly}
 \end{align}
These conditions then lead to constraints on 
$H_{AM}^{(1,2)}$,  $H_{B M}^{(1,2)}$ which we now derive.

\begin{proposition}
For \itsimei\ that satisfies
 Eqs. (\ref{Eq:VABMIntOnly})--(\ref{Eq:Interactiononly}),
to first order in $g$ we obtain the condition that
\begin{align}
H_{AM}^{(1)} &= H_{AM} + K_M \qquad H_{AM}^{(2)} = -H_{AM} + K_M\label{Eq:HAM}\\
H_{BM}^{(1)}  & = H_{BM}-K_M \qquad  H_{BM}^{(2)} = -H_{BM} - K_M\label{Eq:HBM}\\
&\qquad \qquad \quad h_{AB}^{(1)} = 0 
\label{Eq:CatlyticCdt-1}
\end{align}
with $K_M$ acting only on the mediator system. 
\end{proposition}

\begin{proof}
At first order in $g$ we find
\begin{equation}
    V_{ABM} = \id + \ii g ( H_{AM}^{(1)}+H_{AM}^{(2)}+H_{BM}^{(1)}+ H_{BM}^{(2)}) +O(g^2).
\end{equation}
The right hand side is catalytic for the mediator iff $H_{AM}^{(1)}+H_{AM}^{(2)}+H_{BM}^{(1)}+ H_{BM}^{(2)}=- h_{AB}^{(1)}$. 
Since the left hand side is a sum of an operator acting on $AM$ and an operator acting on $BM$, while  
$h_{AB}^{(1)}$ is interaction only on $AB$, we must have $h_{AB}^{(1)}=0$. Therefore we have
$H_{AM}^{(1)}+H_{AM}^{(2)}= - \left( H_{BM}^{(1)}+ H_{BM}^{(2)}\right) $.
Since the left hand side acts on systems $A$
 and $M$, while the right hand side acts on systems $B$ and $M$, equality implies that in fact they only act on system $M$, and we can define
 $K_M :=(H_{AM}^{(1)}+H_{AM}^{(2)})/2= - \left( H_{BM}^{(1)}+ H_{BM}^{(2)}\right)/2$. 
This then implies 
Eqs. (\ref{Eq:HAM}) and (\ref{Eq:HBM}) for  $H_{AM}:= (H_{AM}^{(1)}-H_{AM}^{(2)})/2$ and $H_{BM}:=(H_{BM}^{(1)}-H_{BM}^{(2)})/2$.

\end{proof}

For simplicity of the analysis, from now on we make the further hypothesis that 
 \begin{align}
 H_{AM}^{(1,2)} \text{ and }
  H_{BM}^{(1,2)}
 \text{are interaction only }
\label{Eq:IntonlyHABM}
 \end{align}
 which implies that 
 \begin{equation}
     K_M=0.
 \end{equation}

\begin{proposition}
For \itsimei\ satisfying
Eqs.~(\ref{Eq:VABMIntOnly})--(\ref{Eq:Interactiononly})
and~(\ref{Eq:IntonlyHABM}), catalyticity to second order requires
the relevant mediator operators to have scalar cross-commutators.
A nontrivial interaction $h_{AB}^{(2)}\neq0$ therefore requires an
infinite-dimensional mediator. A canonical realization of these
relations, adopted below but not claimed to be exhaustive, uses a
finite number $N$ of bosonic modes. In this realization the
system--mediator interactions have the form
\begin{align}
H_{AM} &=  \sum_{i=1}^N (\hat A_i \otimes a_i +h.c.) \nonumber\\
 H_{BM} &= \sum_{i=1}^N (\hat B_i \otimes a_i +h.c.),
 \label{Eq:HAMHBM}
\end{align}
where the $a_i, a_i^\dag$ are bosonic annihilation and creation operators acting on the mediator space with $[a_i,a_j^\dag]=\delta_{ij}$, and the $\hat A_i , \hat B_i$ are arbitrary traceless operators acting on systems $A$ and $B$, respectively. 
The effective Hamiltonian to order $g^2$ then has the form
\begin{align}
h_{AB}^{(2)} 
&= 
\ii \sum_{i=1}^N ( \hat A_i \hat B_i^\dag - \hat A_i^\dag \hat B_i) .
\label{Eq:HhAB(2)}
\end{align}
\end{proposition}

\begin{proof}
Setting $K_M$ to zero in Eqs. (\ref{Eq:HAM}), (\ref{Eq:HBM}), inserting this into the left hand side of Eq.  (\ref{Eq:VABMIntOnly}), and expanding 
to second order in $g$, we obtain
\begin{align}
    V_{ABM} (g) 
    &= \id + g^2 [H_{AM}, H_{BM}]+O(g^3).
\end{align}
Hence, at this order,  the interaction is catalytic for the mediator to second order in $g$ if
\begin{equation}\label{eq:commutator}
   \ii \,[H_{AM},H_{BM}] =   h_{AB}^{(2)} \otimes \id_{M}.
\end{equation}

We expand the local Hamiltonians as
\begin{align}
H_{AM} &=  \sum_k A_k \otimes M_k \label{Eq:HAMHBM-z2}\\
H_{BM} &= \sum_\ell B_\ell \otimes \widetilde M_\ell
\label{Eq:HAMHBM-z}
\end{align}
where the operators $A_i,B_i,M_i$ and $\widetilde M_i$ act on systems $A,B$, and $M$, respectively, 
and $\{A_k\}_k$, $\{B_\ell\}_\ell$ constitute an orthonormal basis of traceless operators. (There are no terms proportional to $\id_A$ or $\id_B$ in Eqs. (\ref{Eq:HAMHBM-z2}) and (\ref{Eq:HAMHBM-z}) in view of the hypothesis Eq. (\ref{Eq:IntonlyHABM})).
 The relation ~(\ref{eq:commutator}) then takes the form 
\begin{align}
- \ii  \left(  h_{AB}^{(2)} \otimes \id_{M} \right) =
   &  \sum_{k,\ell } A_k \otimes B_\ell \otimes [M_k,\widetilde M_\ell].
   \label{Eq:h(2)AB}
\end{align}
Since $\{  A_k \otimes B_\ell\}_{k,\ell}$ are linearly independent, 
this identity requires that the commutator of every pair of mediator operators be scalar: $[M_k,\widetilde M_\ell]= c_{k,\ell}\, \id_M$ with $c_{k,\ell}\in\mathds{C}$.
If $M$ is finite dimensional, then we can take the trace of this relation, obtaining that 
all constants $c_{jk} =0$. Thus a nonzero second order interaction requires infinite dimensional mediators. 

The constraint that commutators
$[M_k , \widetilde M_\ell]$  are scalars 
suggests that the mediator can be most simply realized by
a Heisenberg-type algebra~\cite{bourbaki1989lie}.
From now on we make this hypothesis, and suppose that 
the mediator operators can be represented as linear combinations of
independent bosonic creation  and annihilation operators.
For finite-dimensional systems $A$ and $B$,
only a finite number $N$ of bosonic modes are needed. 
 We can thus write
\begin{align}
M_k &=  \sum_{i=1}^N \alpha_{i}^{(k)} a_i + \beta_{i}^{(k)} a_i^\dag
\nonumber
\\
\widetilde  M_\ell &=  \sum_{i=1}^N \alpha_{i}^{(\ell)} a_i + \beta_{i}^{(\ell)} a_i^\dag
\label{Eq:Mk} \\
&{\text{with }} \quad [a_i,a_j]=
[a_i^\dag , a_j^\dag]=0, 
 \quad 
[a_i, a_j^\dag]=\delta_{ij} .
\label{Eq:aadagger}
\end{align}

Furthermore,  using the fact that $H_{AM}$ is Hermitian, we obtain 
\begin{align}
H_{AM} &= \frac{1}{2} \sum_{k} (A_k \otimes M_k + A_k^\dag \otimes M_k^\dag) \\& =
\frac{1}{2} \sum_k \sum_{i=1}^N (A_k \otimes (\alpha_{i}^{(k)} a_i + \beta_{i}^{(k)} a_i^\dag ) 
\nonumber\\
& \quad \quad
+ A_k^\dag \otimes (\bar \alpha_{i}^{(k)} a_i^\dag + \bar \beta_{i}^{(k)} a_i)) \\
& =\sum_{i=1}^N \left(\sum_k  \frac{\alpha_i^{(k)} A_k + \bar \beta_i^{(k)} A_k^\dag}{2}\right)  \otimes a_i + \mathrm{h.c.} \\
& := \sum_{i=1}^N (\hat A_i \otimes a_i +h.c.)
\end{align}
where we have defined new operators $\hat A_i =\frac12 \sum_k  {\alpha_i^{(k)} A_k + \bar \beta_i^{(k)} A_k^\dag}$.
Similarly, we obtain the decomposition $H_{BM}=\sum_i (\hat B_i \otimes a_i +h.c.)$.
The operators $\hat A_i$ and $\hat B_i$ are traceless since the operators $A_k, B_\ell$ are traceless.
Direct computation of the commutator $[H_{AM},H_{BM}] $ using Eq. (\ref{Eq:aadagger})
then yields Eq. (\ref{Eq:HhAB(2)}).
\end{proof}

Under the same assumptions, consider the third order expansion of the propagator
\begin{align}
V_{ABM}(g)&=
   e^{- \ii g H_{BM} } 
 e^{- \ii g  H_{AM} } 
 e^{\ii g H_{BM} } 
 e^{\ii g H_{AM} } \\
 &= \id + g^2 [H_{AM}, H_{BM}]
\nonumber\\
&\quad 
 -\ii \frac{g^3}{2} [H_{AM}+H_{BM},[H_{AM}, H_{BM}]] +
 O(g^4).
\end{align}
Up to third order, state-independent mediation therefore requires
\begin{align}
    \Delta_{ABM}
    &:=[H_{AM}+H_{BM},[H_{AM},H_{BM}]]
    \nonumber\\
    &=-\ii[H_{AM}+H_{BM},h_{AB}^{(2)}]
    =\Delta_{AB}\otimes\id_M ,
\end{align}
where the third-order effective Hamiltonian coefficient is
$h_{AB}^{(3)}=\Delta_{AB}/2$.

Explicitly, using (\ref{Eq:HhAB(2)}), we have
\begin{align}
   \Delta_{ABM}
   =&
   \sum_i a_i \sum_j 
   \left( [\hat A_i, \hat A_j]\hat B_j^\dag
   - [\hat A_i, \hat A_j^\dag] \hat B_j
   \right. \nonumber\\
   & \left.
+ [\hat B_i,\hat B_j^\dag]\hat A_j
- [\hat B_i, \hat B_j]\hat A_j^\dag
   \right)
   +h.c. \ .
\end{align}

Let us now give a sufficient condition for $V_{ABM}(g)$ to correspond to an effective Hamiltonian to all orders in $g$. We do not aim to give minimal conditions, but rather a natural sufficient form.

\begin{proposition}\label{Prop3}
If $V_{ABM}(g)$ is given by 
\begin{align}
 V_{ABM} (g)= e^{- \ii g H_{BM} } 
 e^{- \ii g  H_{AM} } 
 e^{\ii g H_{BM} } 
 e^{\ii g H_{AM} } 
 \label{Eq:HABM-22B}
 \end{align}
with
$H_{AM}$ and $H_{BM}$ given by 
\begin{align}
H_{AM} &= \sum_i (\hat A_i \otimes a_i +h.c.) \nonumber\\
 H_{BM} &= \sum_i (\hat B_i \otimes a_i +h.c.)
 \label{Eq:HAMHBM-3}
\end{align}
and let
\begin{align}
h_{AB}^{(2)}  = \ii \sum_i( \hat A_i \hat B_i^\dag - \hat A_i^\dag \hat B_i) .
\end{align}
Then, if
\begin{align}\label{eq:third-order-catalytic-F}
  [\hat A_j,  h_{AB}^{(2)} ] &= 0 \qquad \forall j, \nonumber\\
   [\hat B_j,  h_{AB}^{(2)} ] &= 0 \qquad \forall j, 
\end{align}
we have
\begin{equation}
V_{ABM}(g)= \exp(-\ii g^2 h_{AB}^{(2)}) \otimes \id_M,
\label{eq:third-order-catalytic-G}
\end{equation}
i.e. the  interaction is given by the Hamiltonian $ H^{\rm eff}_{AB} = g^2h_{AB}^{(2)}$.
\end{proposition}

\begin{proof}
We recall the Baker-Campbell-Hausdorff series of the group commutator~\cite{bourbaki1989lie} 
\begin{align}\label{eq:BCH}
   &  e^{-Y}e^{-X}e^{Y}e^{X} =e^Z \nonumber\\
& \text{with}
    \qquad Z= -[X,Y] + \frac{1}{2}[X+Y,[X,Y]]+\dots,
\end{align}
where the missing terms are higher order commutators of $X$ and $Y$ with $[X,Y]$.

Recall that $h_{AB}^{(2)} =\ii [H_{AM},H_{BM}]  $.
Note that since $ h_{AB}^{(2)} $ is Hermitian,
Eq. (\ref{eq:third-order-catalytic-F})
also 
implies that
$  [\hat A_j^\dag,  h_{AB}^{(2)} ] = [\hat B_j^\dag,  h_{AB}^{(2)} ] = 0$
for all $j$.
Therefore, when
 Eq. (\ref{eq:third-order-catalytic-F}) holds,
all higher nested commutators in the BCH expansion (\ref{eq:BCH})  vanish, thus Eq. (\ref{eq:third-order-catalytic-G}) holds.

\end{proof}

\begin{proposition}
Under the conditions of Prop. \ref{Prop3}, 
in the  case where there is a single mediator mode $N=1$, we have
\begin{align}
H_{AM} &= \sum_i (x_i a + \bar x_i a^\dag) \ketbra{i}_A
\quad x_i \in \mathbb{C} ,
 \label{Eq:N=1HAM}\\
H_{BM}&= \sum_j (y_j a + \bar  y_j a^\dag) \ketbra{j}_B 
\quad y_j \in \mathbb{C} , \label{Eq:N=1HBM} \\
    h_{AB}^{(2)} &=
\sum_{ij}
\ii
\left(
x_i\bar y_j-\bar x_i y_j
\right)
|ij\rangle\langle ij|.
 \label{Eq:N=1hAB2}\\
    V_{ABM} (g)&= \exp(-\ii g^2 h_{AB}^{(2)}) \nonumber\\
    &= \sum_{ij} e^{   g^2 (x_i \bar  y_j-\bar  x_i y_j )} \ketbra{ij}_{AB} \otimes \id_M .\label{Eq:N=1VAB2}
\end{align}
Because $\hat A$ and $\hat B$ are traceless,
$\sum_i x_i=\sum_j y_j=0$.
The simple model in Eq. (\ref{Eq:Interactions}) is obtained when the $x_i$ are real and the $y_j$ are purely imaginary.
\end{proposition}

\begin{proof}

In the case of a single mediator mode, Eq. (\ref{Eq:HAMHBM-3}) implies that $H_{AM}= a \hat A+ a^\dag \hat A^\dag$ and $H_{BM}= a \hat B+ a^\dag \hat B^\dag$.
We therefore have
\[
h_{AB}^{(2)}
=
\ii(\hat A\hat B^\dag-\hat A^\dag\hat B).
\]
The assumptions of Prop.~\ref{Prop3} imply
\[
[\hat A,h_{AB}^{(2)}]
=
-\ii[\hat A,\hat A^\dag]\otimes \hat B=0,
\]
and
\[
[\hat B,h_{AB}^{(2)}]
=
\ii \hat A\otimes[\hat B,\hat B^\dag]=0.
\]
For nontrivial $\hat A$ and $\hat B$, this gives
\[
[\hat A,\hat A^\dag]=0,
\qquad
[\hat B,\hat B^\dag]=0.
\]
Hence $\hat A$ and $\hat B$ are normal and, since they are finite dimensional, can be diagonalized:
\[
\hat A=\sum_i x_i\ketbra{i}_A,
\qquad
\hat B=\sum_j y_j\ketbra{j}_B,
\]
 for some complex coefficients $x_i$ and $y_j$. This implies 
 Eqs. (\ref{Eq:N=1HAM}) and (\ref{Eq:N=1HBM}). Direct computation using Eq. (\ref{Eq:HhAB(2)}) then yields the right hand side of Eq. (\ref{Eq:N=1hAB2}). Since we are in the conditions of  Prop. \ref{Prop3}, we have the right hand side of Eq. (\ref{Eq:N=1VAB2}).
\end{proof}

\section{The continuous field model} \label{app:ContField}

In this appendix we derive the continuous field model. We start with a discrete stroboscopic model and then take the limit of infinitely many modes. 

\subsection{Stroboscopic model}

We consider a sequence of mediators that follow each other at very short intervals, and interact successively with systems $A$, $B$, again $A$, and again $B$. To this end we consider the waveguide loop in Fig.~\ref{Fig:Mediators-Schrodinger} and divide each segment of length $L$ and the flight time $T$ into $n$ elementary steps of duration $\tau=T/n$. As depicted in Fig.~\ref{Fig:Mediators-Schrodinger}, we label 
 each small spatial section by the integer
 index $i$. To each section $i$ we  associate a bosonic mode described by ladder operators  $ a_i,   a_i^\dag$ with $[ a_i,  a_i^\dag]=1$. Only the modes $0$, ${n}$, ${2n}$ and ${3n}$  interact with  systems $A$ and $B$. 
Each time step corresponds to the combination of the 
shift operator
\begin{equation}
K:   a_i\mapsto K   a_i K^\dag=  a_{i+1}
\end{equation}
and the
interactions described by 
\begin{align}
W_0
&:=
U_{AM_0}^{(+)}
U_{BM_{n}}^{(+)}
U_{AM_{2n}}^{(-)}
U_{BM_{3n}}^{(-)} \quad \text{with}\\
U_{AM_i}^{(\pm)} &= \exp(\pm \ii g \, x_i\otimes A)=\exp(\pm \ii \bar g \,( a_i^\dag +  a_i) \otimes A) \label{eq:app-UA}
\\ 
U_{BM_i}^{(\pm)} &= \exp(\pm \ii g \,  p_i\otimes B)= \exp(\pm \ii \bar g \,(\ii   a_i^\dag - \ii  a_i) \otimes B), \label{eq:app-UB}
\end{align}
where $  x_i = \frac{ a_i^\dag+ a_i}{\sqrt{2}}$ and $ p_i = \ii \frac{  a_i^\dag -  a_i}{\sqrt{2}}$ are the canonical position and momentum operators of the mode $i$, and for convenience we have introduced the rescaled interaction  parameter 
\begin{equation}\bar g:=\nicefrac{g}{\sqrt{2}}.\end{equation}
An elementary evolution step is given by the product 
\begin{equation}
U_{step} = K W_0 ,
\end{equation}
so after $N$ time steps the evolution is generated by the propagator $(K W_0 )^N$.

  To further analyze the dynamics, we introduce
 the spectral decompositions of the  operators $A=\sum_a\lambda_a\ketbra{a}_A$, $B=\sum_b\lambda_b'\ketbra{b}_B$. The local interactions can then be expressed as  controlled displacements
\begin{align}
U_{AM_i}^{(\pm)}&=\sum_a \mathcal{D}_{i}\!\left(\pm\ii\lambda_a \bar g\right)\otimes\ketbra{a}_A, \nonumber\\
U_{BM_i}^{(\pm)}&=\sum_b \mathcal{D}_{i}\!\left(\mp \lambda_b' \bar{g}\right)\otimes\ketbra{b}_B, \label{eq:controlled-disp}
\end{align}
where $\mathcal{D}_{i}(\alpha):=\exp(\alpha\,   a_i^\dag-\alpha^*  a_i)$.

The dynamics of the bosonic modes thus consists of hopping along the incoming waveguide and undergoing controlled displacements at sites $i=0,n,2n, 3n$.
 It remains to compute the result of these controlled displacements.
The product of two displacements $\mathcal D(\beta) \mathcal D(\alpha)=e^{\ii \, \operatorname{Im}[\beta \alpha^*]}\mathcal{D} (\alpha+\beta)$ is a displacement times a phase.
One can view the phase ${\rm Im}[\beta \alpha^*] = (\alpha_{\rm Re}\beta_{\rm Im} -\alpha_{\rm Im}\beta_{\rm Re})$ as the geometric phase given by twice the (oriented) area enclosed by the successive displacements, seen as vectors in $\mathds{R}^2$.
We therefore have:
\begin{align}
  \mathcal{D}(\ii \bar g\lambda_a)  &= \mathcal{D}(\ii \bar g\lambda_a)  \nonumber\\
  \mathcal{D}(-\bar g\lambda_b') \mathcal{D}(\ii \bar g\lambda_a)  &=
      e^{\ii \bar g^2 \lambda_a \lambda_b'}\mathcal{D}(\bar g (\ii \lambda_a-\lambda_b'))\nonumber\\
     \mathcal{D}(-\ii \bar g\lambda_a) \mathcal{D}(-\bar g\lambda_b') \mathcal{D}_i(\ii \bar g\lambda_a) &=  e^{2 \ii \bar g^2 \lambda_a \lambda_b'}\mathcal{D}(-\bar g\lambda_b')
     \nonumber\\
  \mathcal{D}(\bar g\lambda_b')  \mathcal{D}(-\ii \bar g\lambda_a) \mathcal{D}(-\bar g\lambda_b') \mathcal{D}_i(\ii \bar g\lambda_a)
  &=    e^{\ii 2 \bar g^2 \lambda_a \lambda_b'} 
  \label{Eq:RelDisp}
\end{align}
In the last line of Eq. (\ref{Eq:RelDisp}), the
 sequence of four controlled displacements traces a rectangle in phase space whose enclosed area produces a pure phase.
 Thus, after the bosonic mode has exited the system following the fourth displacement, it returns to its original state, leaving only a phase corresponding to an interaction between systems $A$ and $B$.
The evolution is thus catalytic for the mediator.
Taking into account that the duration of an elementary step is $ T  /n$, 
we can reinterpret the phase $e^{\ii 2 \bar g^2 \lambda_a \lambda_b'} $ as corresponding to 
 the effective interaction (using $2\bar g^2 = g^2$)
\begin{equation}
 H_{AB}^{\rm eff}  = -\frac{n g^2}{T} A B .
\label{Eq:Heff}
\end{equation}

Let us suppose that the mediators all arrive in the same state $\ket{\phi_0}$, and denote the product
$\ket{\Phi_0}_M=\bigotimes_i \ket{\phi_0}_{M_i}$.
Let us also suppose that the systems $A$ and $B$ are in eigenstates $\ket{ a}\otimes \ket{b}$ of the coupling operators. Then
the following state (we include phases for easier reference to Eq. (\ref{Eq:RelDisp}))
\begin{align}
    \ket{\Psi^0_{ab}}_{ABM} &=
  \left( \prod_{i=1}^{n}     \mathcal{D}_i(\ii \bar g\lambda_a)   \right)
\left( 
\prod_{i=n+1}^{2n}
 e^{\ii \bar g^2 \lambda_a \lambda_b'}\mathcal{D}_i(\bar g (\ii \lambda_a-\lambda_b'))  \right)
\left(  \prod_{i=2n+1}^{3n}
       e^{2 \ii \bar g^2 \lambda_a \lambda_b'}\mathcal{D}_i(-\bar g\lambda_b') \right)
       \ket{\Phi_0}_M \otimes \ket{a} \otimes \ket b
       \label{Eq:Stationary}
      \end{align} 
is stationary:
\begin{align}
   KW_0     \ket{\Psi^0_{ab}}_{ABM} & =      e^{2 \ii \bar g^2 \lambda_a \lambda_b'} 
       \ket{\Psi^0_{ab}}_{ABM} .
\end{align}
This is the exact dressed state corresponding to an eigenstate of the stroboscopic dynamics, corresponding to the eigenstate $\ket{a,b}$ of the Hamiltonian in Eq.~(\ref{Eq:Heff}) in the effective picture.

A particularly simple case is when the mediators arrive in the vacuum state $\ket 0$. A displacement acting on the vacuum state is a coherent state denoted 
$ \mathcal{D} (\alpha)  \ket 0 =: \ket \alpha $. Then, the stationary state reads
\begin{align}
    \ket{\Psi^0_{ab}}_{ABM} &=
     e^{\ii 3 n \bar g^2 \lambda_a \lambda_b'}
  \left(    \prod_{i \leq 0}     \ket{0}_i \right)
\left( \prod_{i=1}^{n}
\ket{ \ii \bar g\lambda_a}_i  \right)
\left( 
\prod_{i=n+1}^{2n}
\ket{ \bar g (\ii \lambda_a-\lambda_b') }_i  \right)
\left(  \prod_{i=2n+1}^{3n}
  \ket{ -\bar g\lambda_b'}_i \right)
 \left(   \prod_{i>3n}     \ket{0}_i   \right) \otimes \ket{a} \otimes \ket b .
       \label{Eq:Stationary-B}
      \end{align} 

\subsection{The continuum limit}

The continuum limit, depicted in Fig.~\ref{Fig:ABM-FF}(b), is obtained by letting the number of discrete modes $n$ in each section of the waveguide tend to infinity. We see already from Eq.~(\ref{Eq:Heff}) that a finite effective coupling energy $E_{ab}= -\mu \lambda_a \lambda_b'$ is obtained by increasing $n$ while scaling
 $g=\sqrt{\frac{\mu T}{n}}$ with $\mu$ fixed.

Let us introduce the spatial coordinate $x$ along the waveguide of length $3L$, and the infinitesimal length $\dd x = \frac{L}{n}$ and the infinitesimal  time $\dd t = \frac{T}{n}$. The mode index $i$ can then be put in correspondence with the spatial coordinate $x=\frac{i}{n} L$. This allows us to define the creation and annihilation field operators  $ a^\dag(x) = \frac{ a_i^\dag}{\sqrt{\dd x}}$  and $ a(x) = \frac{ a_i}{\sqrt{\dd x}}$ satisfying the commutation relation 
\begin{equation}
    [ a(x),  a^\dag(x')] = \frac{\delta_{ii'}}{\dd x} \to \delta(x-x').
\end{equation}
For instance the total number of bosons in a waveguide segment is then given by the operator
\begin{equation}
    \sum_{i=1}^{n}  a_i^\dag  a_i = \sum_{i=1}^{n}{\dd x} \frac{  a_i^\dag   a_i}{\dd x}   \to \int_0^L \dd x \,  a^\dag(x)  a(x).
\end{equation}

For the interaction parameter $g=\sqrt{\frac{\mu T}{n}} = \sqrt{\mu \, \dd t}$ the elementary interactions (Eqs.~\ref{eq:app-UA} and \ref{eq:app-UB}) of each discrete mode with the systems $A$ or $B$ have duration $\dd t$ and can be written as
\begin{align}
  U_{AM_i}^{(\pm)} &= \exp(\pm \ii g \,\frac{ a_i^\dag +  a_i}{\sqrt{2}} \otimes A)  = \exp(\pm \ii \sqrt{\frac{\mu \dd t \dd x}{2 }} \,( a^\dag(x) + a(x)) \otimes A) \\&= \exp(\pm \ii \, \kappa \, \dd x \,( a^\dag(x) + a(x)) \otimes A)
      = \exp(\pm \ii \ v \kappa \, \dd t \,( a^\dag(x) + a(x)) \otimes A)\\
     U_{BM_i}^{(\pm)} &= \exp(\mp \kappa \, \dd x \,( a^\dag(x) -  a (x)) \otimes B)=\exp(\mp v \kappa \, \dd t \,(a^\dag(x) -  a(x)) \otimes B).
\end{align}
Here  $\kappa := \sqrt{\frac{\mu}{2 v}}$ is the coupling constant and $v= \dd x/\dd t = L/ T$  the velocity at which the field travels through the waveguide. The instantaneous infinitesimal local displacements of the field at locations $x \in\{0,L,2L,3L\}$ become
\begin{align} 
\mathcal{D}_{0} &= e^{\ii \kappa \lambda_a \dd x(a^\dag(0)+  a(0)) } &\qquad \mathcal{D}_{2L} &= e^{-\ii \kappa \lambda_a \dd x(  a^\dag(2L)+  a(2L)) }
\\
\mathcal{D}_{L} &= e^{- \kappa \lambda_b' \dd x(  a^\dag(L)-  a(L)) } &\qquad  \mathcal{D}_{3L} &= e^{ \kappa \lambda_b' \dd x(  a^\dag(3L)-  a(3L)) }.
\end{align}
The global elementary interactions $W_0=U_{AM_0}^{(+)} U_{BM_{n}}^{(+)} U_{AM_{2n}}^{(-)}
U_{BM_{3n}}^{(-)} $ can now be seen as a continuous dynamics $W_0=\exp(-\ii \dd t (H_{A\bm M}+H_{B\bm M}))$ generated by the local Hamiltonians
\begin{align}
H_{A\bm M} &= -v\kappa \big(  a(0)-  a(2L)+\mathrm{h.c.}\big)\otimes A\qquad 
\nonumber\\ \text{and}\qquad
 H_{B\bm M} &=  v\kappa \big(\ii   a(L)-\ii   a(3L)+\mathrm{h.c.}\big)\otimes B.
\end{align}

In addition, the field is propagating through the waveguide with velocity $v$, as given by the translation operator $T_t   a(x) T_t^\dag =    a(x+v t)$. Therefore, for a time $0<t\leq T$ the evolution of the system is given by the following propagator 
\begin{align}\label{eq:field-propagator}
    V_t = &\sum_{a,b}\ketbra{a,b}_{AB}  \nonumber\\ 
    &\otimes  e^{\ii \kappa \lambda_a \int_{0}^{vt} \dd x\left(   a^\dag (x)+  a(x)\right) - \kappa \lambda_b' \int_{L}^{L+vt}  \dd x \left(   a^\dag (x)-  a(x)\right) -\ii \kappa \lambda_a \int_{2L}^{2L+vt} \dd x \left(   a^\dag (x)+  a(x)\right) +\kappa \lambda_b' \int_{3L}^{3L+vt}  \dd x \left(   a^\dag (x)-  a(x)\right)} T_t\ .
\end{align}
Let us now take the limit $n\to \infty$ in Eq.~(\ref{Eq:Stationary-B}) to obtain the stationary states in the continuum picture when the initial state of the mediator is the vacuum state:
\begin{align}
 \ket{\Psi_{ab}}_{AB M} &= 
    \ket{a,b}_{AB}  \ket{\mathfrak{F}_{ab}}_{ M} \nonumber\\   
 \ket{\mathfrak{F}_{ab}}_{  M}&= 
 \exp( \ii \nicefrac{3}{2} \mu T \lambda_a \lambda_b' )
  \exp(\kappa \int_0^{3L} \dd x\Big(   a^\dag(x) \zeta_{ab}[x] - \text{h.c.}\Big) )\ket{0}_{  M} \nonumber\\
 &=   \exp ( i \nicefrac{3}{2} \mu T \lambda_a \lambda_b' )
    {\exp(-\frac{\kappa^2}{2} \int_0^{3L}  \dd x  \,
    |\zeta_{ab}[x]|^2 )} \exp(\kappa \int_0^{3L} \dd x \,   a^\dag(x) \zeta_{ab}[x] )\ket{0}_{ M}.
\end{align}
where
\begin{align} \label{eq:amplitude-disp}
    \zeta_{ab}[x] := \begin{cases}
      \ii \lambda_a & 0< x\leq L \\
        \ii \lambda_a - \lambda_b' & L< x\leq 2L \\
       -  \lambda_b' & 2 L< x\leq 3 L \\
        0 & \text{otherwise}
   \end{cases} 
\end{align}

Finally, in the Heisenberg picture the freely translated field is
$T_t^\dagger a(x)T_t=a(x-vt)$ and therefore obeys the right-moving
convective derivative
$\partial_+a(x,t):=(\partial_t+v\,\partial_x)a(x,t)$. Including its
interaction with systems $A$ and $B$ gives the Heisenberg equation
\begin{equation}
   \partial_+a(x,t)
   =\ii[H_{A\bm M}(t)+H_{B\bm M}(t),a(x,t)].
\end{equation}

\subsection{Effect of a fast local gate}
\label{App:ContLocalGate}

We consider a system that is initially, for $t<0$, in a superposition of stationary states
\begin{equation}
     \ket{\Psi}_{AB  M} = \sum_{a,b} c_{ab}\,  e^{\ii \mu \lambda_{a} \lambda_b' t}  \ket{a,b}_{AB} \ket{\mathfrak{F}_{ab}}_{ M}.
\end{equation}
At $t=0$, an operator $U_A= \sum_{a',a} \mathcal{A}_{a'a} \ketbra{a'}{a}_A$ is applied to system $A$.
The change propagates through the waveguide at velocity $v$.
After a duration $3T$ the system settles down to a superposition of new stationary states
$ \ket{\Psi_{a'b}}_{AB M} $ each evolving with its phase $e^{\ii \mu \lambda_{a'} \lambda_b' t} $. 
The dressing of the state does not change instantaneously. Rather, the difference between the two dressed states is emitted as a traveling coherent pulse of duration $2T$ which carries information about the transition $\ket{a}\to \ket{a'}$ and gives rise to decoherence of the $AB$ system; see Fig.~\ref{Fig:ABM-Flux}(b). Let us describe these effects in detail.

At $t=0_+$, immediately after the action of $U_A$, the state becomes
\begin{equation}
     \ket{\Psi'}_{AB  M} =  \sum_{a',a,b} \mathcal{A}_{a'a} c_{ab} \ket{a',b}_{AB} \ket{\mathfrak{F}_{ab}}_{  M}.
\end{equation}
To understand the evolution of this state, we consider each component $\ket{a',b}_{AB} \ket{\mathfrak{F}_{ab}}_{\bm M}$ separately. Notice that for $a\neq a'$ this state is not stationary. Instead, its evolution 
can be obtained
by multiplying with  the propagator Eq.~(\ref{eq:field-propagator}):
\begin{equation}
\ket{a',b}_{AB} \ket{\mathfrak{F}_{ab}^{a'}(t)}_{ M}:=
    V_t \ket{a',b}_{AB} \ket{\mathfrak{F}_{ab}}_{ M} .
\end{equation}
 For $t>0$, the displacement of the field at locations $x\in\{0,2L\}$ is controlled by the value $\lambda_{a'}$. Thus, for $0< t \leq 2T$, at $x=2L$ the infinitesimal displacement of the field mode $a(2L)$ by $-\ii \kappa \lambda_{a'} \dd x$ no longer cancels the displacement of $\ii \kappa \lambda_a \dd x$ applied to the same mode at $t-2 T$. 

Concretely, to describe the state of the field after the application of the gate let us introduce the compact notation
\begin{align}
  \ket{\alpha}_{M[x_0,x_1]} := \exp(\int_{x_0}^{x_1}\dd x (\alpha  a^\dag(x)- \alpha^*   a(x)))\ket{0}_{M[x_0,x_1]} = e^{-|\alpha|^2\frac{x_1-x_0}{2}}\exp(\int_{x_0}^{x_1}\dd x \, \alpha  a^\dag(x))\ket{0}_{M[x_0,x_1]} ,
\end{align}
describing coherent fields in the coordinate interval $[x_0,x_1]$. At $t=0$ the state of the field is 
\begin{align}
     \ket{\mathfrak{F}_{ab}^{a'}(0)}_{M}  = &e^{-\ii \nicefrac{3}{2} E_{ab} T} 
        \ket{\ii\kappa \lambda_{a}}_{ M[0,L]}
        \ket{\kappa(\ii \lambda_a - \lambda_b')}_{M[L ,2L]}
        \ket{-\kappa\lambda_b'}_{ M[2L,3L]}
        \ket{0}_{ M[0,3L]^c},
\end{align}
where $E_{ab} = -\mu \lambda_a \lambda_b' $
and $\ket{0}_{\bm M[0,3L]^c}$ denotes the vacuum state outside the interval $[0,3L]$.

For $0<t\leq T$ the state of the field is found to be
\begin{align}
    \ket{\mathfrak{F}_{ab}^{a'}(t)}_{ M} = &e^{-\ii \nicefrac{3}{2} E_{ab}T} e^{-\ii \frac{E_{ab}+E_{a'b}}{2} t}
        \ket{\ii\kappa \lambda_{a'}}_{M[0,v t]}
        \ket{\ii\kappa \lambda_a}_{\bm M[v t,L]}
        \ket{\kappa(\ii \lambda_a - \lambda_b')}_{M[L ,2L]}
        \nonumber\\
        &\ket{\kappa(\ii \Delta \lambda_{aa'}-\lambda_b')}_{ M[2L,2L+v t]}
        \ket{-\kappa\lambda_b'}_{ M[2L+v t,3L]}
        \ket{0}_{ M[0,3L]^c},
\end{align}
where $\Delta \lambda_{aa'} =  \lambda_a-\lambda_{a'}$.
At $t=T$ this simplifies to the following 
\begin{align}
    \ket{\mathfrak{F}_{ab}^{a'}(T)}_{ M} = &e^{-\ii \nicefrac{3}{2}  E_{ab} T}  e^{-\ii \frac{E_{ab}+E_{a'b}}{2} T}
        \ket{\ii\kappa \lambda_{a'}}_{ M[0,L]}
        \ket{\kappa(\ii \lambda_a - \lambda_b')}_{M[L ,2L]}
        \ket{\kappa(\ii \Delta \lambda_{aa'}-\lambda_b')}_{M[2L,3L]}
        \ket{0}_{ M[0,3L]^c}.
\end{align}
For later times $T<t \leq 2T$ the influence of the transition $\ket{a}\to \ket{a'}$ reaches the locations $x\in\{L,3L\}$, and we find 
\begin{align}
\ket{\mathfrak{F}_{ab}^{a'}(2T)}_{ M}
={}&e^{-\ii\frac{3}{2}E_{ab}T}
e^{-\ii2TE_{a'b}}
\ket{\ii\kappa\lambda_{a'}}_{ M[0,L]}
\ket{\kappa(\ii\lambda_{a'}-\lambda_b')}_{M[L,2L]}
\nonumber\\
&\otimes
\ket{\kappa(\ii\Delta\lambda_{aa'}-\lambda_b')}_{ M[2L,3L]}
\ket{\ii\kappa\Delta\lambda_{aa'}}_{ M[3L,4L]}
\ket{0}_{ M[0,4L]^c}.
\end{align}
Finally, at time $t=3T$ all field modes inside the waveguide have been controlled by the state $\ket{a'}$ of system $A$, leading to
\begin{align}
 \ket{\mathfrak{F}_{ab}^{a'}(3T)}_{ M}
 &=e^{-\ii T(E_{ab}+2E_{a'b})}
   \ket{\mathfrak{F}_{a'b}}_{ M}
   \ket{\ii\kappa\Delta\lambda_{aa'}}_{M[3L,5L]}
   \ket{0}_{ M[0,5L]^c}
   \nonumber\\
 &=e^{\ii\mu T\Delta\lambda_{aa'}\lambda_b'}
   e^{-\ii3TE_{a'b}}
   \ket{\mathfrak{F}_{a'b}}_{\bm M}
   \ket{\ii\kappa\Delta\lambda_{aa'}}_{ M[3L,5L]}
   \ket{0}_{ M[0,5L]^c}.\label{eq app: state final cont perp}
\end{align}

For all $t\geq3T$, the waveguide and systems $A$ and $B$ are again in the stationary state $\ket{\mathfrak{F}_{a'b}}_{M} $, which will continue to evolve with the energy $E_{a'b}$. However, the field $\ket{\ii\kappa \Delta \lambda_{aa'}}_{ M[3L,5L]}$ that left the interaction region carries away some information about the transition $\ket{a}\to \ket{a'}$ that was implemented at time $t=0$. In addition, note that the phase factor appearing in Eq.\eqref{eq app: state final cont perp} has two contributions. The phase $-3TE_{a'b}$ is the usual accumulated phase during $3T$ associated to the energy of the state $\ket{a',b}$; this phase will continue accumulating at later times. The phase $\mu T\Delta\lambda_{aa'}\lambda_b'$ can be understood as a Lamb-shift-like correction to the usual phase, due to the fact that the energies of the system $AB$ were perturbed during the time $3T$ where the stationary states were perturbed. 

The radiated state $\ket{\ii\kappa \Delta \lambda_{aa'}}_{ M[3L,5L]}$ is the vacuum displaced by the difference between the eigenvalues $\lambda_a-\lambda_{a'}$. These states satisfy 
\begin{equation}\braket{\ii\kappa \Delta \lambda_{aa'}}{\ii\kappa \Delta \lambda_{\bar a \bar a'}}_{ M[3L,5L]}=e^{-\frac{\mu T}{2}(\Delta \lambda_{aa'}-\Delta \lambda_{\bar a\bar a'})^2}.
\end{equation}
Thus, if we trace out the field radiated outside of the interaction region and ignore the deterministic static Lamb-shift-like corrections to the phases of the states $\ket{a',b}$, in the effective picture the action of the unitary $U_A$ must be replaced by the transformation
\begin{align}
    \mathcal{U}_A[\cdot] = \!\!\!\!\sum_{a',a, \bar a ,\bar a'}\!\!\!\! e^{-\frac{\mu T}{2}(\Delta \lambda_{aa'}-\Delta \lambda_{\bar a\bar a'})^2} \mathcal A_{a'a}  \ketbra{a'}{a} \cdot  \ketbra{\bar a}{\bar a'}  \mathcal A^*_{\bar a' \bar a},
\end{align}
applied at the same time $t=0$ on the system $A$. This channel describes the unitary $U_A$ that is dephased with respect to the eigenbasis of the coupling observable $A$. This dephasing is due to the radiated field recording the change $\lambda_a-\lambda_{a'}$ of the corresponding quantity realized by $U_A$.

As a simple example in this interaction picture, consider a qubit with $A=\frac{1}{2}Z$ and a local bit-flip $u_A = X$. There are only two possible transitions $\ket 0 \to \ket 1$ and $\ket 1 \to \ket 0$, with the spin-$Z$ changes $\Delta\lambda_{01}=+1$ and $\Delta\lambda_{10}=-1$, respectively.
Hence, the implemented gate is not a unitary bit flip, but rather 
\begin{equation}
\mathcal{X}_A[\rho]= X
\left(\begin{array}{cc}
\rho_{00} & e^{-2 \mu T}\rho_{01} \\
e^{-2 \mu T}\rho_{10} &\rho_{11}
\end{array}\right)
X \end{equation}
expressed in the computational basis. 

\section{Effect of local gates} \label{app:pert}

We consider here the situation where, in addition to the mediator that implements the effective interaction between $A$ and $B$, system $A$ also evolves independently.  Following the discussion in the main text let the local dynamics of Alice's system (in the absence of the mediator) between times $0$ and $t$ be given by a time-dependent unitary $u_A(t)$.  

\subsection{ Sudden evolution of system A}

In the single mediator loop picture, let $A$ undergo a local unitary evolution $U_A= u_A(t+2T)u_A(t)^\dag$ between its two interactions with the mediator at times $t$ and $t+2T$. The overall evolution of the systems over the mediation loop (between $t$ and $t+4T$) is then given by the propagator
\begin{align}
V_{ABM}' &=
u_A(t+4T)u_A^\dag(t+2T) e^{-\ii gBp}\,e^{-\ii gA x} u_A(t+2T)u_A(t)^\dag \,e^{\ii gBp}\,e^{\ii gAx} u_A(t) \\
&= u_A(t+4T)\left( e^{-\ii gBp}\,e^{-\ii gA_{t+2T} x}  \,e^{\ii gBp}\,e^{\ii gA_tx}\right),
\label{eq:app-disturbed-loop-2A}
\end{align}
where we defined the rotating observable 
\begin{equation}
A_t := u_A^\dag(t) A u_A(t).
\end{equation}
For simplicity assume that $u_A(t)$ is constant before $t$ and after $t+2T$ to rewrite 
\begin{align}
V'_{ABM}
&= U_A e^{-\ii gBp}\,e^{-\ii gA_1x}\,e^{\ii gBp}\,e^{\ii gA_0x}
\label{eq:app-disturbed-loop-2B} \quad \text{with} \quad A_0 := A_t \quad , \quad A_1 := U_A^{\dag} A_t U_A,
\end{align}
and analyze the operator
\begin{align}
\tilde V_{ABM} 
&= e^{-\ii gBp}\,e^{-\ii gA_1x}\,e^{\ii gBp}\,e^{\ii gA_0x}.
 \label{Eq:tildeVABM}
\end{align}

The action of $V'_{ABM}$ is in general not catalytic for the mediator: after the interaction, the mediator and the systems are left in an entangled state, corresponding to nonunitary reduced evolution of the composite system $AB$.

\subsection{Decoherence dependence on initial mediator state}

Before considering the perturbative regime, it is instructive to analyze two limiting mediator states exactly. These illustrate how decoherence depends on the mediator initial state independently of any small-$g$ approximation.
We contrast two initial states of the mediator: an eigenstate of position $\vert x_0 \rangle_M$ or an eigenstate of momentum $\vert p_0 \rangle_M$.

Writing $A_0=\sum_{a_0} \lambda_{a_0} \ketbra{a_0}$, 
$A_1=\sum_{a_1} \lambda_{a_1} U_A^\dag \ketbra{ a_1} U_A =\sum_{a_1} \lambda_{a_1} \ketbra{\tilde a_1}$, and 
$B=\sum_b \lambda_{b}' \ketbra{b}$, in the first case we find
\begin{equation}
\tilde V_{ABM} \vert x_0 \rangle_M
= \left(\sum_{a_0, a_1, b}  
e^{-ig(\lambda_{a_1} - \lambda_{a_0}) x_0} e^{  i g^2 \lambda_{a_1} \lambda_{b}'  }\braket{\tilde  a_1}{a_0}  \, 
\ketbra{\tilde  a_1}{a_0}_A\otimes\ketbra{b}_B 
\right) 
 \vert x_0 \rangle_M,
 \label{eq:app-disturbed-loop-3}
\end{equation}
where we have again used the identity $\mathcal D(\beta) \mathcal D(\alpha)=e^{\ii \, {\rm Im}[\beta \alpha^*]}\mathcal{D} (\alpha+\beta)$.
Crucially, the mediator returns to its initial state, while, for $A_0 \neq A_1$, the term in parentheses produces a residual $O(g)$ unitary transformation of the systems $AB$ (Lamb-shift-like correction), which depends on the value of $x_0$ via the phase $e^{-ig(\lambda_{a_1} - \lambda_{a_0}) x_0}$. Hence, for such initial states, the evolution is unitary and catalytic for the mediator, although which evolution is implemented depends on the initial position $x_0$ of the mediator.

On the other hand, when the mediator enters in the momentum eigenstate, we have
\begin{equation}
\tilde V_{ABM} \vert p_0 \rangle_M
=\sum_{a_0, a_1, b}  \left(
e^{\ii g^2 \lambda_{a_1} \lambda_{b}'}\ \braket{\tilde a_1}{a_0}
\ketbra{\tilde a_1}{a_0}_A \otimes\ketbra{b}_B 
 \otimes 
 \vert p_0 + g(\lambda_{a_0} - \lambda_{a_1}) \rangle_M\right).
  \label{eq:app-disturbed-loop-4}
\end{equation}
The final state of the mediator is modified if the eigenvalues $\lambda_{a_0}\neq \lambda_{a_1}$. In this case the mediator becomes entangled with system $A$, and decoherence occurs.

Note that if it is $B$ which undergoes a sudden unitary evolution, then the roles of $x$ and $p$ and the initial states of the mediator are reversed. Note also that an intermediate situation will occur if the mediator is in a state with uncertainty in both $x$ and $p$. In the next section, we study, perturbatively in $g$, the case where the mediator is in the vacuum state, which ensures symmetry between local evolutions of $A$ and $B$.

\subsection{Small interaction strength}
\label{Sub-SmallIInt}

Let us now suppose that the interaction strength $g$ is small (we do not suppose that $U_A$ is small). 
Expanding to second order in $g$, one finds
\begin{align}
\tilde V_{ABM}
&=
\id
-\ii g (A_1-A_0)x
+g^2[A_1x,Bp]
-\frac{g^2}{2}\Big((A_1x)^2+(A_0x)^2-2A_1A_0x^2\Big)
+O(g^3).
\end{align}
Using $[x,p]=\ii$ and the fact that $A_1$ and $B$ act on different systems, this becomes
\begin{align}
\tilde V_{ABM}
&=
\id
-\ii g\,\Delta A\,x
+\ii g^2 A_1B
-\frac{g^2}{2}\Big(\Delta A^2 - [A_1,A_0]\Big)x^2 
+O(g^3), 
\label{eq:app-V-expand}
\end{align}
where $\Delta A = A_1 - A_0$.

In Eq. (\ref{eq:app-V-expand}) we see that the effective interaction between systems $A$ and $B$ is  given by $A_1 B$, taking into account the evolution of the system $A$ between mediator interactions. We also find that entanglement is generated between the system $A$ and the mediator through the terms in $\Delta A x$ and $[A_1,A_0]x^2$. Note that both effects arise from the expansion to order  $g^2$  of the terms already written out in Eqs. (\ref{eq:app-disturbed-loop-3}) and (\ref{eq:app-disturbed-loop-4}).

We now assume that the mediator is initially in the vacuum state $\ket{0}_M$, with $\bra 0 x \ket 0 =0,\bra 0 x^2 \ket 0 =1/2 $. 
This is a very natural choice, since (i) it ensures that the uncertainties in $x$ and $p$ are equal, and therefore that systems $A$ and $B$ are treated similarly (all results would also map to local evolution of system $B$), and (ii) it has vanishing expectation values for both $x$ and $p$, ensuring that there is no first order correction to the interaction Hamiltonian due to the term $-\ii g\,\Delta A \langle x\rangle$.

The reduced evolution resulting from $\tilde V_{ABM}$ is therefore given by
\begin{equation} \label{eq: mapping of state}
\rho_{AB} \mapsto \tr_M \qty[ \tilde V_{ABM} \, \rho_{AB} \otimes \ketbra{0}{0}_M \tilde  V_{ABM}^\dag ]  =  \rho_{AB} + g^2 \qty( \ii[A_1B-H_A^{\rm L},\rho_{AB}]
+\frac12\mathcal D_{\Delta A}(\rho_{AB})) + O(g^3),
\end{equation}
where 
\begin{equation}
H_A^{\rm L}
:=\frac{\ii}{4}[A_1,A_0] \qquad \text{and} \qquad \mathcal D_{\Delta A}(\rho)
=
\Delta A\,\rho\,\Delta A
-\frac12\{\Delta A^2,\rho\}.
\end{equation}
Here, $H_A^{\rm L}$ is a correction to the Hamiltonian (analogous to a Lamb shift) and $\mathcal D_{\Delta A}$ is a Lindblad dissipator generating decoherence.

The local evolution therefore gives rise to a residual coupling with the mediator through the mismatch $\Delta A$, and produces dissipation which only vanishes when $A_1=A_{0}$.

\subsection{Continuous control of system A}

We now turn to the continuous-control limit of Fig.~\ref{Fig:ABM-Flux}(a), in which $A$ undergoes continuous local evolution, generated by the control Hamiltonian $h_A(t)$, which is slowly-varying and weak on the mediator flight timescale. As in the main text, we work in the interaction picture of this control, with $h_A^{\rm I}(t):=u_A^\dag(t)h_A(t)u_A(t)$. Specifically, we assume that $h_A(t)\approx h_A(t+2T)$ and 
\begin{align}
    \Delta A_t &= u_A^\dag(t+2T) A  \,u_A(t+2T)- u_A^\dag(t) A\,  u_A(t) \approx \ii  2 T [h_A^{\rm I}(t),A_t] =: 2 T \dot A_t \\
    \frac{\ii }{4} [A_{t+2T},A_t] & \approx  \, \frac{T}{2} [A_t, [h_A^{\rm I}(t),A_t]] =  -\ii \frac{T}{2} [A_t,\dot A_t].
\end{align}
Writing $\rho_{AB}$ for the interaction-picture state, we can now rewrite Eq.~\eqref{eq: mapping of state} as
\begin{equation}\label{state trans}
\rho_{AB} \mapsto  \rho_{AB} +g^2 \qty( \ii[A_{t+2T} B-H_A^{\rm L},\rho_{AB}]
+2\mathcal D_{T\dot A_t}(\rho_{AB})) + O(g^3),
\end{equation}

Considering the  non-overlapping case, illustrated in Fig.~\ref{Fig:ABM-Flux}(a), where new mediators are injected every $\tau = 4 T$, we note that the state transformation in Eq.\eqref{state trans} takes the time $4T$. Thus assuming that this time is short and defining the effective interaction strength $\mu =\frac{g^2}{4T}$, we can rewrite it as a Lindblad master equation with 
\begin{align}
\frac{\dd }{\dd t}\rho_{AB} = \ii \mu [A_{t+2T} B-H_A^{\rm L}, \rho_{AB}]+2\mu \, \mathcal D_{T\dot A_t}(\rho_{AB})
\end{align}
where 
\begin{equation}
H_A^{\rm L}
:=-\ii \frac{T}{2} [A_t,\dot A_t] \qquad \text{and} \qquad \mathcal D_{T\dot A_t}(\rho)
=
T^2(\, \dot A_t\,\rho\,\dot A_t
-\frac12 \{\dot A_t^2,\rho\}).
\end{equation}

We see that the dissipation rate in the master equation is given by
\begin{equation}
    r =  2 \mu  T^2 \tr \left[(\dot A_t)^2 \rho_{AB}\right] =  2 \mu T^2 \tr \left[-[h_A^{\rm I}(t),A_t]^2 \rho_{AB}\right].
\end{equation}
It is quadratic in the strength of the control Hamiltonian, while the time $t_{\rm gate}$ needed to implement a desired local gate on the system $A$, is inversely proportional to the strength of the control. Taking this into account, we find that the accumulated effect of decoherence during the gate-time is given by 
\begin{align}
    t_{\rm gate} r = t_{\rm gate} \times 2 \mu T^2 \tr \left[-[h_A^{\rm I}(t),A_t]^2 \rho_{AB}\right] \sim \frac{T^2 \mu }{t_{\rm gate}};
\end{align}
it can be neglected in the limit where local control satisfies $t_{\rm gate}\gg T^2 \mu$.

\section{Local exchange interactions}
\label{app:JC-dressed}

We now derive the dressed states of the local Jaynes--Cummings model in the
stroboscopic geometry of Fig.~\ref{Fig:strobo}.  We take the stroboscopic
period to be $\tau=T$ and use the standard convention in which $g$ is the
dimensionless coupling of each physical atom--mediator encounter.  Define 
\begin{equation}
 H_{XM}(a):=\sigma_X^{(+)}a -\sigma_X^{(-)}a^\dagger,
 \qquad
 \sigma_X^{(+)}=\ketbra{1}{0}_X,
 \quad
 \sigma_X^{(-)}=\ketbra{0}{1}_X,
 \label{eq:JC-local-generator}
\end{equation}
where $a,a^\dag$ are the ladder operators of the mediator modes coupled to the system $X$. The purpose of the calculation is twofold.  First, we must determine the system states that admit a stationary $O(g)$ field dressing in the presence of the mediators. Second, only after including this dressing can we calculate the $O(g^2)$ eigenphase and decide whether the corresponding quasienergy has an imaginary part.  

At a reference stroboscopic time, let $a_0,a_1,a_2,a_3$ denote the incoming
mode and the three modes already in flight.  The two pairs of local
interactions occurring during one step can then be written as
\begin{align}
 U_{A M}(g)
 &=\exp\!\left\{g\left[H_{AM}(a_0)-H_{AM}(a_2)\right]\right\},
 \nonumber\\
 U_{B M}(g)
 &=\exp\!\left\{g\left[H_{BM}(a_1)-H_{BM}(a_3)\right]\right\}.
 \label{eq:JC-strobo-unitaries}
\end{align}
Equation~(\ref{eq:JC-strobo-unitaries}) specifies the microscopic timing
convention: the two couplings acting on the same atom at a stroboscopic time
are represented by one exponential of their summed generators.  Replacing
this simultaneous coupling by two ordered exponentials would add a
same-atom commutator at order $g^2$ and define a different model.

Let $\widetilde K$ denote the unitary translation of the complete mediator
train, $\widetilde K a_j^\dagger\widetilde K^\dagger=a_{j+1}^\dagger$, and let
$\Pi_{\rm vac}$ project onto the sector with no excitation in the outgoing
mode $a_4$.  We write $K:=\Pi_{\rm vac}\widetilde K$ for the shift projected onto the no outgoing excitation sector.
In the one-excitation sector we therefore have
\begin{equation}
K a_j^\dagger\ket{0}_{ M}=a_{j+1}^\dagger\ket{0}_{ M}
\quad (j=0,1,2),
\qquad
K a_3^\dagger\ket{0}_{ M}=0.
\label{eq:JC-active-shift}
\end{equation}
 Accordingly, the
one-period no-emission map is the contraction
\begin{equation}
 F_\tau(g)=K U_{B M}(g)U_{A M}(g).
 \label{eq:JC-one-period-map}
\end{equation}
We seek to determine the eigenstates
\begin{equation}
F_\tau\ket{\Psi}=\lambda\ket{\Psi} \quad , \quad \lambda \in \mathbb{C}
\label{eq:Ftau}
\end{equation}
 Then $1-|\lambda|^2$ is the probability
lost from the active no-emission branch in one period.

\subsection{Perturbative eigenproblem}

Introduce
\begin{align}
 G_A&=H_{AM}(a_0)-H_{AM}(a_2),
 &G_B&=H_{BM}(a_1)-H_{BM}(a_3),
 \nonumber\\
 G&=G_A+G_B,
 &G^{(2)}&=\frac{G_A^2}{2}+\frac{G_B^2}{2}+G_BG_A .
 \label{eq:JC-G-definitions}
\end{align}
The ordered product in Eq.~(\ref{eq:JC-one-period-map}) has the expansion
\begin{equation}
 F_\tau(g)=K\left[\id+gG+g^2G^{(2)}+O(g^3)\right].
 \label{eq:JC-map-expansion}
\end{equation}
The term $G_BG_A$ in $G^{(2)}$ is important: it retains the ordering of
the $A$ and $B$ interactions during one stroboscopic step.

For an atomic state $\ket{\varphi_\alpha}_{AB}$, we seek a dressed eigenstate and
its eigenvalue in the form
\begin{align}
 \ket{\Psi_\alpha}
 &=\ket{\varphi_\alpha}_{AB}\ket{0}_{ M}
   +g\ket{\chi_\alpha^{(1)}}+O(g^2),
 \nonumber\\
 \lambda_\alpha
 &=1+g^2\lambda_\alpha^{(2)}+O(g^3).
 \label{eq:JC-perturbative-ansatz}
\end{align}
Let $P=\id_{AB}\otimes\ketbra{0}{0}_{ M}$ be the projector onto the mediator vacuum, and choose the perturbative
gauge $P\ket{\chi_\alpha^{(1)}}=0$.

There is no first-order eigenvalue correction because
$PGP=0$: acting once on an atomic state and the mediator vacuum creates one
mediator excitation and is therefore orthogonal to the vacuum sector.

Inserting this ansatz into Eq. (\ref{eq:Ftau}) gives the first-order equation
\begin{equation}
 (K-\id)\ket{\chi_\alpha^{(1)}}
 =-KG\ket{\varphi_\alpha,0},
 \label{eq:JC-first-order-equation}
\end{equation}
whereas projection of the second-order equation onto the mediator vacuum
gives
\begin{equation}
 PK\left(G^{(2)}\ket{\varphi_\alpha,0}
       +G\ket{\chi_\alpha^{(1)}}\right)
 =\lambda_\alpha^{(2)}\ket{\varphi_\alpha,0}.
 \label{eq:JC-second-order-equation}
\end{equation}
To obtain the second equation, one first keeps the unknown term
$g^2\ket{\chi_\alpha^{(2)}}$ on both sides of the eigenvalue equation and then
projects with $P$.  Since translation leaves the mediator vacuum invariant,
the projected $\ket{\chi_\alpha^{(2)}}$ terms cancel.  Equations
(\ref{eq:JC-first-order-equation}) and
(\ref{eq:JC-second-order-equation}) therefore determine the $O(g)$ dressing
and the $O(g^2)$ phase without requiring the full second-order eigenvector.

\subsection{Zero- and single-excitation sectors}

Since $G\ket{00,0}=0$, the vacuum state
\begin{equation}
 \ket{\Psi_{00}}=\ket{00}_{AB}\ket{0}_{ M}
 \label{eq:JC-exact-vacuum}
\end{equation}
is an exact stationary state for every $g$.

In the single-excitation sector, write
$\ket{\varphi}_{AB}=c_A\ket{10}_{AB}+c_B\ket{01}_{AB}$.  
The first order equation will not determine $c_A$ and $c_B$, but will determine the $O(g)$ dressed state $\ket{\chi^{(1)}(c_A,c_B)}$ as a function of $c_A$ and $c_B$. Using
Eq.~(\ref{eq:JC-active-shift}), the source term in the first-order equation is
\begin{equation}
-KG\ket{\varphi,0}
=\ket{00}_{AB}\left[c_A(a_1^\dagger-a_3^\dagger)
+c_Ba_2^\dagger\right]\ket{0}_{ M}.
\label{eq:JC-single-source}
\end{equation}
Using 
$(K-\id)(a_1^\dagger+a_2^\dagger)\ket0
=(a_3^\dagger-a_1^\dagger)\ket0$ and
$(K-\id)(a_2^\dagger+a_3^\dagger)\ket0=-a_2^\dagger\ket0$, one finds that the solution of 
Eq.~(\ref{eq:JC-first-order-equation}) is
\begin{equation}
 \ket{\chi^{(1)}(c_A,c_B)}
 =-\ket{00}_{AB}
 \left[c_A(a_1^\dagger+a_2^\dagger)
       +c_B(a_2^\dagger+a_3^\dagger)\right]\ket{0}_{ M}.
 \label{eq:JC-single-first-order}
\end{equation}

We next compute
\begin{align}
PKG^{(2)}\ket{\varphi,0}
&=-\left(c_A\ket{10}+c_B\ket{01}\right)\ket0,
\nonumber\\
PKG\ket{\chi^{(1)}(c_A,c_B)}
&=\left[(c_A+c_B)\ket{10}+(-c_A+c_B)\ket{01}\right]\ket0.
\label{eq:JC-single-second-order-parts}
\end{align}
Adding them cancels the diagonal terms and leaves
\begin{equation}
 PK\left(G^{(2)}\ket{\varphi,0}
       +G\ket{\chi^{(1)}(c_A,c_B)}\right)
 =\left(c_B\ket{10}-c_A\ket{01}\right)\ket{0}_{ M}.
 \label{eq:JC-single-second-order}
\end{equation}
In the ordered atomic basis $\{\ket{10},\ket{01}\}$, Eq. (\ref{eq:JC-second-order-equation}) is the eigenvalue equation for
the anti-Hermitian matrix
$\left(\begin{smallmatrix}0&1\\-1&0\end{smallmatrix}\right)$.  Hence its
eigenvalues are purely imaginary: there is an $O(g^2)$ phase but no
$O(g^2)$ loss. The eigenvectors and eigenvalues of this anti-Hermitian matrix are
\begin{equation}
 \ket{\psi_\pm}=\frac{\ket{10}\pm\ii\ket{01}}{\sqrt{2}},
 \qquad
 \lambda_\pm^{(2)}=\pm\ii .
 \label{eq:JC-single-eigenvectors}
\end{equation}
Consequently, at a reference stroboscopic time, the eigenstates and eigenvalues are
\begin{align}
 \ket{\Psi_\pm}
 ={}&\ket{\psi_\pm}_{AB}\ket{0}_{ M}
 \nonumber\\
 &-\frac{g}{\sqrt{2}}\ket{00}_{AB}
 \left[(a_1^\dagger+a_2^\dagger)
 \pm\ii(a_2^\dagger+a_3^\dagger)\right]\ket{0}_{ M}
 +O(g^2),
 \label{eq:JC-dressed-pm}\\
 \lambda_\pm
 ={}&1\pm\ii g^2+O(g^3)
 =e^{\pm\ii g^2}+O(g^3).
 \label{eq:JC-phase-pm}
\end{align}
After $k$ periods, the leading phase is therefore
$e^{\pm\ii k g^2}$.  Equivalently, writing
$\lambda_\pm=e^{-\ii E_\pm\tau}$ gives
$E_\pm=\mp g^2/\tau=\mp\mu$, in agreement with the eigenvalues of
Eq.~(\ref{eq:HJCeff-B}).

\subsection{Two-excitation sector}

For the unperturbed state $\ket{11}_{AB}\ket{0}_{ M}$,
the source term in the first-order equation is
\begin{equation}
-KG\ket{11,0}
=\left[a_2^\dagger\ket{10}_{AB}
+(a_1^\dagger-a_3^\dagger)\ket{01}_{AB}\right]\ket0_{M}.
\label{eq:JC-double-source}
\end{equation}
Eq.~(\ref{eq:JC-first-order-equation}) therefore gives
\begin{equation}
 \ket{\chi_{11}^{(1)}}
 =-\left[(a_2^\dagger+a_3^\dagger)\ket{10}_{AB}
        +(a_1^\dagger+a_2^\dagger)\ket{01}_{AB}\right]
   \ket{0}_{ M}.
 \label{eq:JC-double-first-order}
\end{equation}
At second order, the two terms on the left hand side of Eq. (\ref{eq:JC-second-order-equation}) cancel:
\begin{align}
 PKG^{(2)}\ket{11,0}&=-2\ket{11,0},
 \nonumber\\
 PKG\ket{\chi_{11}^{(1)}}&=+2\ket{11,0}.
 \label{eq:JC-double-cancellation}
\end{align}
It follows that $\lambda_{11}^{(2)}=0$ and hence
\begin{align}
 \ket{\Psi_{11}}
 ={}&\ket{11}_{AB}\ket{0}_{ M}
 \nonumber\\
 &-g\left[(a_2^\dagger+a_3^\dagger)\ket{10}_{AB}
         +(a_1^\dagger+a_2^\dagger)\ket{01}_{AB}\right]
   \ket{0}_{ M}+O(g^2),
 \label{eq:JC-dressed-11}\\
 \lambda_{11}={}&1+O(g^3).
 \label{eq:JC-phase-11}
\end{align}
The $O(g^2)$ terms implicit in Eqs.~(\ref{eq:JC-dressed-pm}) and
(\ref{eq:JC-dressed-11}) include the corresponding normalization
corrections as well as genuine $O(g^2)$ corrections to the dressed state.\\

We leave open the extension of the  JC model to the continuous field limit.

\end{document}